\documentclass[10pt,a4paper,dvipsnames,usenames]{article}
\usepackage{amssymb}
\usepackage{amsmath}
\usepackage{amsthm}
\usepackage{latexsym}
\usepackage[dvips]{epsfig}
\usepackage{mathrsfs}
\usepackage{eufrak}
\usepackage{bm}
\usepackage{stmaryrd}
\usepackage{authblk}
\usepackage{slashed}
\usepackage{bigints}

\usepackage[dvipsnames]{xcolor}
\usepackage{tikz}
\theoremstyle{plain}
\newtheorem{proposition}{Proposition}
\newtheorem{lemma}{Lemma}
\newtheorem{theorem}{Theorem}

\newtheorem{corollary}{Corollary}
\newtheorem*{main}{Theorem}

\newtheorem{remark}{Remark}

\def\bma{{\bm a}}
\def\bmb{{\bm b}}
\def\bmc{{\bm c}}
\def\bmd{{\bm d}}
\def\bme{{\bm e}}
\def\bmf{{\bm f}}
\def\bmg{{\bm g}}
\def\bmh{{\bm h}}
\def\bmi{{\bm i}}
\def\bmj{{\bm j}}
\def\bmk{{\bm k}}
\def\bml{{\bm l}}
\def\bmn{{\bm n}}

\def\bmu{{\bm u}}
\def\bmv{{\bm v}}

\def\bmx{{\bm x}}

\def\bmzero{{\bm 0}}
\def\bmone{{\bm 1}}
\def\bmtwo{{\bm 2}}
\def\bmthree{{\bm 3}}

\def\bmK{{\bm K}}
\def\bmL{{\bm L}}

\def\bmbeta{{\bm \beta}}

\def\bmxi{{\bm \xi}}

\def\bmomega{{\bm \omega}}

\def\bmphi{{\bm \phi}}

\def\bmsigma{{\bm \sigma}}

\def\bmupsilon{{\bm \upsilon}}
\def\bmell{{\bm \ell}}

\def\bmGamma{{\bm \Gamma}}

\def\bmpartial{{\bm \partial}}
\def\bmnabla{{\bm \nabla}}

\def\bmell{{\bm \ell}}

\usepackage{mathtools}% http://ctan.org/pkg/mathtools
\usepackage{xparse}% http://ctan.org/pkg/xparse
\makeatletter
\newcommand{\raisemath}[1]{\mathpalette{\raisem@th{#1}}}
\newcommand{\raisem@th}[3]{\raisebox{#1}{$#2#3$}}
\makeatother
\NewDocumentCommand{\newrbar}{O{0pt} O{0pt}}{
  \ensuremath{\mathrlap{\raisemath{#2}{\hspace*{#1}{\mathchar'26\mkern-9mu}}}r}}

\newcounter{mnotecount}%[section]

\newcommand{\mnotex}[1]%{}
{\protect{\stepcounter{mnotecount}}$^{\mbox{\footnotesize $\bullet$\themnotecount}}$ 
\marginpar{%\color{red}%
\raggedright\tiny\em
$\!\!\!\!\!\!\,\bullet$\themnotecount: #1} }

\newcounter{mnote}

\usepackage{color}
\usepackage{xcolor}
\usepackage[normalem]{ulem}
\usepackage{soul}
\usepackage{hyperref}

\begin{document}

\title{\textbf{On the non-linear stability of the Cosmological region
    of the Schwarzschild-de Sitter spacetime}}
 
\author[1]{ Marica Minucci  \footnote{E-mail
    address:{\tt m.minucci@qmul.ac.uk}}}
\author[1]{Juan A. Valiente Kroon \footnote{E-mail address:{\tt j.a.valiente-kroon@qmul.ac.uk}}}

\affil[1]{School of Mathematical Sciences, Queen Mary, University of London,
Mile End Road, London E1 4NS, United Kingdom.}

\maketitle
\begin{abstract}
The non-linear stability of the sub-extremal Schwarzschild-de Sitter
spacetime in the stationary region near the conformal boundary is
analysed using a technique based on the extended conformal Einstein
field equations and a conformal Gaussian gauge. This strategy relies
on the observation that the Cosmological stationary region of this
exact solution can be covered by a non-intersecting congruence of
conformal geodesics. Thus, the future domain of dependence of suitable
spacelike hypersurfaces in the Cosmological region of the spacetime
can be expressed in terms of a conformal Gaussian gauge. A
perturbative argument then allows us to prove existence and stability
results close to the conformal boundary and away from the asymptotic
points where the Cosmological horizon intersects the conformal
boundary. In particular, we show that small enough perturbations of
initial data for the sub-extremal Schwarzschild-de Sitter spacetime give
rise to a solution to the Einstein field equations which is regular at
the conformal boundary. The analysis in this article can be regarded
as a first step towards a stability argument for perturbation data on
the Cosmological horizons.
\end{abstract}

\section{Introduction}
One of the key problems in mathematical General Relativity is that of
the non-linear stability of black hole spacetimes. This problem is challenging for
its mathematical and physical features. Most efforts to establish the
non-linear stability of black hole spacetimes in both the
asymptotically flat and Cosmological setting have, so far, relied on
the use of vector field methods ---see e.g. \cite{DafRod10}.

\medskip
The results in \cite{Fri86b,Fri91,CFEBook} show that
the \textit{conformal Einstein field equations} are a
powerful tool for the analysis of the stability of vacuum
asymptotically simple spacetimes. They provide a system of field
equations for geometric objects defined on a four-dimensional
Lorentzian manifold $(\mathcal{M}, \bmg)$, the so-called
\textit{unphysical spacetime}, which is conformally related to a
spacetime $(\tilde{\mathcal{M}}, \tilde{\bmg})$, the so-called
\textit{physical spacetime}, satisfying the Einstein field
equations. The usefulness of the conformal transformation relies on
the fact that global problems for the physical spacetimes are recasted
as local existence problems for the unphysical spacetime. The conformal Einstein field equations constitute a system of
differential conditions on the curvature tensors with respect to the
Levi-Civita connection of $\bmg$ and the conformal factor $\Xi$. The
original formulation of the equations, see e.g. \cite{Fri84}, requires the introduction of
so-called \textit{gauge source functions} to construct evolution equations. An alternative approach to
gauge fixing is to adapt the analysis to a congruence of curves. A
natural candidate for a congruence is given by \textit{conformal
geodesics} ---a conformally invariant generalisation of the standard
notion of geodesics. Using these curves to fix the gauge allows to
define a \emph{conformal Gaussian system}. To combine this gauge choice with the conformal Einstein
field equations it is necessary to make use of a more general version
of the latter ---the \textit{extended conformal Einstein field
equations}.  The extended conformal field equations have been used to
obtain an alternative proof of the semiglobal non-linear stability of
the Minkowski spacetime and of the global non-linear stability of the
de Sitter spacetime ---see \cite{LueVal09}.

\medskip
In view of the success of conformal methods to analyse the global
properties of asymptotically simple spacetimes, it is natural to ask
whether a similar strategy can be used to study the non-linear
stability of black hole spacetimes. This article gives a first step in
this direction by analysing certain aspects of the conformal structure
of the sub-extremal Schwarzschild-de Sitter spacetime which can be used,
in turn, to adapt techniques from the asymptotically simple setting to
the black hole case. 

\medskip
\noindent
\textbf{The Schwarzschild-de Sitter spacetime.} The Schwarzschild-de Sitter spacetime is a spherically symmetric
solution to the vacuum Einstein field equations with Cosmological
constant. This spacetime depends on the de Sitter-like value of the
Cosmological constant $\lambda$ and on the mass $m$ of the black
hole. Assuming spherical symmetry almost completely singles out the
Schwarzschild-de Sitter spacetimes among the vacuum solutions to the
Einstein field equations with de Sitter-like Cosmological
constant. The other admissible solution is the Nariai spacetime
---see e.g. \cite{Sta98}. In
the Schwarzschild-de Sitter spacetime, the relation between the mass and
the Cosmological constant determines the location of the
\textit{Cosmological} and \textit{black hole horizons} ---see e.g. \cite{GriPod09}.

\smallskip
The Schwarzschild-de Sitter spacetime solution can be studied by means
of the extended conformal Einstein field equations ---see \cite{GasVal17a}. This is in fact a
spacetime with a smooth conformal extension towards the future (or
past). Since the
Cosmological constant takes a de Sitter-like value, the conformal
boundary of the spacetime is spacelike and moreover, there exists a
conformal representation in which the induced $3$-metric on the
conformal boundary $\mathscr{I}$ is homogeneous. Thus, it is possible
to integrate the extended conformal field equations along single
conformal geodesics ---see \cite{GarGasVal18}.

\smallskip
In this article, we analyse the sub-extremal Schwarzschild-de Sitter
spacetime as a solution to the extended conformal Einstein field
equations and use the insights to prove existence and stability
results.

\medskip
\noindent
\textbf{The main result. } The metric of the Schwarzschild-de Sitter
spacetime can be expressed in standard coordinates by the line element
\begin{equation}
\label{BackgroundMetricIntro}
\mathring{\tilde{\bmg}} =-\bigg{(} 1-\frac{2m}{r}-\frac{\lambda}{3}r^2\bigg{)}\mathbf{d} t  \otimes \mathbf{d} t +
\bigg{(}1-\frac{2m}{r}-\frac{\lambda}{3}r^2\bigg{)}^{-1} \mathbf{d} r \otimes  \mathbf{d} r + r^2  \bm{\sigma}.
\end{equation}
\emph{In this article we restrict our attention to a choice of the
  parameters $\lambda$ and $m$ for which the exact solution is sub-extremal} ---see Section
\ref{Section:BackgroundSolution} for a definition of this notion. The
sub-extremal Schwarzschild-de Sitter spacetime has three horizons. Of
particular interest for our analysis is the Cosmological horizon which
bounds a region (\emph{the Cosmological region}) of the spacetime in which the roles of the coordinates
$t$ and $r$ reversed. In analogy to the de Sitter spacetime, the
Cosmological region has an asymptotic region admitting a smooth conformal
extension with spacelike conformal boundary. \emph{In the following, our
analysis will be solely concerned with the Cosmological region.}

\smallskip
The analysis of the conformal properties of the Schwarzschild-de Sitter
spacetime allows us to formulate a result concerning the existence of
solutions to the initial value problem for the Einstein field equations
with de Sitter-like cosmological constant which can be regarded as
perturbations of portions of the initial hypersurface at
$\mathcal{S}_\star\equiv\{r=r_\star \}$ in the Cosmological region of the
spacetime. In this region these hypersurfaces are spacelike and the
coordinate $t$ is spatial. In the following, let ${\mathcal{R}}_\bullet$
denote finite cylinder within $\mathcal{S}_\star$ for which
$|t|<t_\bullet$ for some suitable positive constant $t_\bullet$. Let
$D^+({\mathcal{R}}_\bullet)$ denote the future domain of dependence
of ${\mathcal{R}}_\bullet$. For the Schwarzschild-de Sitter
spacetime such a region is unbounded towards the future and admits a
smooth conformal extension with a spacelike conformal boundary. 

\smallskip
Our main result can be stated as:

\begin{main}
Given smooth initial data $(\tilde{\bmh},\tilde{\bmK})$ for the vacuum Einstein field
equations on ${\mathcal{R}}_\bullet\subset \mathcal{S}_\star$ which is 
suitably close (as measured by a suitable Sobolev norm) to the data
implied by the metric \eqref{BackgroundMetricIntro} in the
Cosmological region of the spacetime, there exists a smooth metric
 $\tilde{\bmg}$ defined over the whole of $D^+({\mathcal{R}}_\bullet)$ which is
close to $\mathring{\tilde{\bmg}}$, solves the vacuum Einstein field equations with
positive Cosmological constant and whose restriction to
${\mathcal{R}}_\bullet$ implies the initial data $(\tilde{\bmh},\tilde{\bmK})$. The metric $\tilde{\bmg}$ admits a
smooth conformal extension which includes a spacelike conformal boundary.
\end{main}

A detailed version of this theorem will be given in Section \ref{Section:ExistenceUniquenesStability}. 

\medskip
Observe that the above result is restricted to the future domain of
dependence of a suitable portion ${\mathcal{R}}_\bullet$ of the
spacelike hypersurface $\mathcal{S}_\star$. The reason for this
restriction is the degeneracy of the conformal structure at the
asymptotic points of the Schwarzschild-de Sitter spacetime where the
conformal boundary, the Cosmological horizon and the singularity seem
to ``meet'' ---see \cite{GasVal17a}. In particular, at these points
the background solution experiences a divergence of the Weyl
curvature. This singularity is remarkably similar to that
produced by the ADM mass at spatial infinity in asymptotically flat
spacetimes ---see e.g. \cite{CFEBook}, chapter 20. It is thus
conceivable that an approach analogous to that used in the analysis of
the problem  of spatial infinity in \cite{Fri98a} may be of help to deal with this
singular behaviours of the conformal structure.

\smallskip
The ultimate aim of the programme started in this article is to obtain
a proof of the stability of the Schwarzschild-de Sitter spacetime for
data prescribed on the Cosmological horizon. Key to this end is the
observation that the hypersurfaces of constant coordinate $r$,
$\mathcal{S}_\star$, can be chosen to be arbitrarily close to the
horizon. As such, an adaptation of the \emph{optimal} local existence
results for the characteristic initial value problem developed in
\cite{Luk12} ---see also \cite{HilValZha20b}--- should allow to
evolve from the Cosmological horizon to a hypersurface
$\mathcal{S}_\star$. These ideas will be developed in a subsequent article. 

\smallskip
It should be stressed that the spacetimes obtained as a result of our
perturbative analysis are dynamic ---in the sense that, generically,
they will not have Killing vectors. This is a consequence of the
fact that initial data sets for the Einstein field
equations admitting solutions to the Killing initial data (KID) equations
are non-generic ---see e.g. \cite{BeiChrSch05}. Whether it is possible to use conformal
Gaussian systems to describe more generic, dynamic, black hole
spacetimes (in both the asymptotically flat and Cosmological setting)
is an interesting and challenging open question which would benefit
from the input of numerical simulations. 

\medskip
\noindent
\textbf{Other approaches.} The non-linear stability of the Schwarzschild-de Sitter spacetime has
been studied by means of the vector field methods that have proven
successful in the analysis of asymptotically flat black holes
---see e.g. \cite{Sch13,Sch15,Sch22}. An alternative approach has made
use of methods of microlocal analysis in the steps of Melrose's school
of geometric scattering ---see \cite{Hin18, HinVas18}. The methods
developed in the present article aim at providing a complementary
approach to the non-linear stability of this Cosmological black hole
spacetime. The interrelation between the results obtained in this
article and those obtained by vector field methods and microlocal
analysis will be discussed elsewhere.

\subsection*{Outline of the article}
This article is organised as follows. In Section
\ref{Section:ConformalGeometry} we provide a succinct discussion of the
tools of conformal geometry that will be used in our analysis  ---the
extended conformal Einstein equations and conformal
geodesics. Moreover, it also discusses the notion of a conformal
Gaussian gauge and provides a hyperbolic reduction of the extended
conformal equations in terms of this type of gauge. Section
\ref{Section:BackgroundSolution} summarises the general properties of the
Schwarzschild-de Sitter spacetime that will be used in our
constructions. Section \ref{Section:CGSdS} describes the construction
of a suitable conformal Gaussian gauge system starting from data
prescribed on hypersurfaces of constant coordinate $r$ on the
Cosmological region of the Schwarzschild-de Sitter spacetime. Section
\ref{Section:SdSGaussian} provides a discussion of the key properties of the
Schwarzschild-de Sitter spacetime in the conformal Gaussian gauge of
Section \ref{Section:CGSdS}. The main existence and stability results
of this article are presented in Section
\ref{Section:ExistenceUniquenesStability}. We conclude the article
with some conclusions and outlook in Section \ref{Section:Conclusions}.

\subsection*{Notations and conventions}
In what follows, the low-case Latin letters $a,\,b,\,c\ldots$ will denote spacetime abstract
tensorial indices, while $i,\,j,\,k,\ldots $ are
spatial tensorial indices ranging from 1 to 3. By contrast, the
low-case Greek letters $\mu,\,
\nu,\,\lambda,\ldots$ and  $\alpha,\, \beta,\gamma,\ldots$ will
correspond, respectively, to spacetime and spatial coordinate
indices. Boldface Latin letters $\bma,\,\bmb,\bmc,\ldots$ will be used
as frame indices.

The signature convention for spacetime metrics is $(-,+,+,+)$. Thus, the
induced metrics on spacelike hypersurfaces are positive definite.

An index-free notation will be often used. Given a 1-form ${\bmomega}$
and a vector ${\bmv}$, we denote the action of ${\bmomega}$ on
${\bmv}$ by $\langle {\bmomega},{\bmv}\rangle$. Furthermore,
${\bmomega}^\sharp$ and ${\bmv}^\flat$ denote, respectively, the
contravariant version of ${\bmomega}$ and the covariant version of
${\bmv}$ (raising and lowering of indices) with respect to a given
Lorentzian metric.  This notation can be extended to tensors of higher
rank (raising and lowering of all the tensorial indices).

The conventions for the curvature tensors will be fixed by the relation
\[
(\nabla_a \nabla_b -\nabla_b \nabla_a) v^c = R^c{}_{dab} v^d.
\]

\section{Tools of conformal geometry}
\label{Section:ConformalGeometry}

The purpose of this section is to provide a brief summary of the
technical tools of conformal geometry that will be used in the
analysis of the stability of the Cosmological region of the
Schwarzschild-de Sitter spacetime. Full details and proofs can be
found in \cite{CFEBook}. 

\subsection{The extended conformal Einstein field equations}
\label{Section:XCFE}
The main technical tool of this article are the \emph{extended
  conformal Einstein field equations} ---see \cite{Fri95,Fri98a}; also
\cite{CFEBook}. This system of equations constitute a conformal
representation of the vacuum Einstein field equations written in terms
of \emph{Weyl connections}. These field equations are formally
regular at the conformal boundary. Moreover, a solution to the extended conformal
equations implies, in turn, a solution to the vacuum Einstein field equations
away from the conformal boundary.  In this section, we provide a brief discussion of 
this system geared towards the applications of this article. A derivation and further discussion of the
general properties of these equations can be found in \cite{CFEBook},
Chapter 8.

\medskip
Throughout this article let $(\tilde{\mathcal{M}},\tilde{\bmg})$ with
$\tilde{\mathcal{M}}$ a 4-dimensional manifold and $\tilde{\bmg}$ a
Lorentzian metric denote a vacuum spacetime satisfying the Einstein field equations with
Cosmological constant
\begin{equation}
\tilde{R}_{ab} =\lambda \tilde{g}_{ab}.
\label{EFE}
\end{equation}
Let $\bmg$ denote an unphysical Lorentzian metric conformally related
to $\tilde\bmg$ via the relation
\[
\bmg = \Xi^2 \tilde\bmg
\]
with $\Xi$ a suitable conformal factor. Let $\nabla_a$ and
$\tilde{\nabla}_a$ denote, respectively, the Levi-Civita connections
of the metrics $\bmg$ and $\tilde\bmg$. The set of points for which
$\Xi=0$ is called the \emph{conformal boundary}.

\subsubsection{Weyl connections}
A Weyl connection is a torsion-free connection $\hat{\nabla}_a$ such
that
\[
\hat{\nabla}_a g_{bc} =-2 f_a g_{bc}.
\]
It follows from the above that the connections $\nabla_a$ and
$\hat{\nabla}_a$ are related to each other by
\begin{equation}
\hat{\nabla}_av^b-\nabla_av^b = S_{ac}{}^{bd}f_dv^c, \qquad
S_{ac}{}^{bd}\equiv \delta_a{}^b\delta_c{}^d +
\delta_a{}^d\delta_c{}^b-g_{ac}g^{bd},
\label{WeylToUnphysical}
\end{equation}
where $f_a$ is a fixed smooth covector and $v^a$ is an arbitrary
vector.  Given that 
\[
\nabla_a v^b -\tilde{\nabla}_a v^b =  S_{ac}{}^{bd} (\Xi^{-1} \nabla_a\Xi)v^c,
\]
one has that 
\[
\hat{\nabla}_av^b -\tilde{\nabla}_av^b =  S_{ac}{}^{bd}\beta_d v^c,
\qquad \beta_d \equiv f_d +\Xi^{-1}\nabla_d\Xi.
\]
In the following, it will be convenient to define
\begin{equation}
d_a \equiv \Xi f_a + \nabla_a \Xi. 
\label{Definition:CovectorD}
\end{equation}

\medskip
In the following $\hat{R}^a{}_{bcd}$ and $\hat{L}_{ab}$ will denote,
respectively, the Riemann tensor and Schouten tensor of the Weyl
connection $\hat{\nabla}_a$. Observe that for a generic Weyl
connection one has that $\hat{L}_{ab}\neq \hat{L}_{ba}$. One has the
decomposition
\[
\hat{R}{}^c{}_{dab} = 2 S_{d[a}{}^{ce}\hat{L}_{b]e} + C^c{}_{dab},
\] 
where $C^c{}_{dab}$ denotes the conformally invariant \emph{Weyl
  tensor}. The (vanishing) torsion of $\hat{\nabla}_a$ will be denoted by
${\Sigma}_a{}^c{}_b$. In the context of the conformal
Einstein field equations it is convenient to define the \emph{rescaled
  Weyl tensor} $d^c{}_{dab}$ via the relation
\[
d^c{}_{dab} \equiv \Xi^{-1} C^c{}_{dab}.
\]

\subsubsection{A frame formalism}
Let $\{ \bme_\bma
\}$, $\bma=\bmzero,\ldots,\bmthree$ denote a $\bmg$-orthogonal frame
with associated coframe $\{ \bmomega^\bma \}$. Thus, one has that
\[
\bmg(\bme_\bma,\bme_\bmb)=\eta_{\bma\bmb}, \qquad \langle
  \bmomega^\bma,\bme_\bmb\rangle =\delta_\bmb{}^\bma.
\]
Given a vector $v^a$, its components with respect to the frame $\{ \bme_\bma
\}$ are denoted by $v^\bma$. Let
  $\Gamma_\bma{}^\bmc{}_\bmb$ and $\hat{\Gamma}_\bma{}^\bmc{}_\bmb$
  denote, respectively, the connection coefficients of $\nabla_\bma$
  and $\hat{\nabla}_a$ with respect to the frame $\{ \bme_\bma \}$. It
  follows then from equation \eqref{WeylToUnphysical} that
\[
\hat{\Gamma}_\bma{}^\bmc{}_\bmb = \Gamma_\bma{}^\bmc{}_\bmd + S_{\bma\bmb}{}^{\bmc\bmd}f_\bmd.
\]
In particular, one has that 
\[
f_\bma = \frac{1}{4}\hat{\Gamma}_\bma{}^\bmb{}_\bmb.
\]
Denoting by $\partial_\bma\equiv \bme_\bma{}^\mu\partial_\mu$ the
directional partial derivative in the direction of $\bme_\bma$, it
follows then that
\begin{eqnarray*}
&& \nabla_\bma T^\bmb{}_\bmc \equiv  e_\bma{}^a \omega^\bmb{}_b\omega^\bmc{}_c(\nabla_a
   T^b{}_c),\\
&&\phantom{\nabla_\bma T^\bmb{}_\bmc}=\partial_\bma T^\bmb{}_\bmc +
   \Gamma_\bma{}^\bmb{}_\bmd T^\bmd{}_\bmc -\Gamma_\bma{}^\bmd{}_\bmc T^\bmb{}_\bmd,
\end{eqnarray*}
with the natural extensions for higher rank tensors and other
covariant derivatives.

\subsubsection{The frame version of the extended conformal Einstein
  field equations}
\label{Section:FrameXCFE}

In this article, we will make use of a frame version of the \emph{extended
conformal Einstein field equations}. In order to formulate these equations
it is convenient to define the following \emph{zero-quantities}:
\begin{subequations}
\begin{eqnarray} 
&& {\Sigma}{}_\bma{}^\bmc{}_\bmb \bme_\bmc\equiv [\bme_\bma, \bme_\bmb] - (\hat{\Gamma}{}_\bma{}^\bmc{}_\bmb- \hat{\Gamma}{}_\bmb{}^\bmc{}_\bma)e_\bmc, \label{ecfe1}\\
&&{\Xi}{}^\bmc{}_{\bmd\bma\bmb}\equiv {R}{}^\bmc{}_{\bmd\bma\bmb} - {\rho}{}^\bmc{}_{\bmd\bma\bmb}, \label{ecfe2} \\
&&{\Delta}{}_{\bmc\bmd\bmb} \equiv \hat\nabla_\bmc \hat{L}{}_{\bmd\bmb} - \hat\nabla_\bmd
   \hat{L}{}_{\bmc\bmb} - d_\bma d{}^\bma{}_{\bmb\bmc\bmd}, \label{ecfe3} \\
&&\Lambda{}_{\bmb\bmc\bmd} \equiv \hat\nabla_\bma
  d^\bma{}_{\bmb\bmc\bmd}-f_\bma d{}^\bma{}_{\bmb \bmc \bmd}, \label{ecfe4}
\end{eqnarray}
\end{subequations}
where the components of the \emph{geometric curvature} ${R}{}^\bmc{}_{\bmd\bma\bmb}$ and the
\emph{algebraic curvature} ${\rho}{}^\bmc{}_{\bmd\bma\bmb}$ are given, respectively, by
\begin{eqnarray*}
&& {R}{}^\bmc{}_{\bmd\bma\bmb} \equiv  \partial_\bma (\hat{\Gamma}{}_\bmb{}^\bmc{}_\bmd)- \partial_\bmb (\hat{\Gamma}{}_\bma{}^\bmc{}_\bmd) + \hat{\Gamma}{}_\bmf{}^\bmc{}_\bmd(\hat{\Gamma}{}_\bmb{}^\bmf{}_\bma - \hat{\Gamma}{}_\bma{}^\bmf{}_\bmb) + \hat{\Gamma}{}_\bmb{}^\bmf{}_\bmd \hat{\Gamma}{}_\bma{}^\bmc{}_\bmf - \hat{\Gamma}{}_\bma{}^\bmf{}_\bmd \hat{\Gamma}{}_\bmb{}^\bmc{}_\bmf,\\
&& {\rho}{}^\bmc{}_{\bmd\bma\bmb} \equiv \Xi {d}{}^\bmc{}_{\bmd\bma\bmb} + 2 {S}{}_{\bmd
   [\bma}{}^{\bmc\bme}\hat{ L}{}_{\bmb]\bme}, 
\end{eqnarray*}
where $\hat{L}_{\bma\bmb}$ and $d^\bmc{}_{\bmd\bma\bmb}$ denote,
respectively, the components of the Schouten tensor of
$\hat{\nabla}_a$ and the rescaled Weyl tensor with respect to the
frame $\{ \bme_\bma \}$. In terms of the zero-quantities \eqref{ecfe1}-\eqref{ecfe4}, the
\emph{extended vacuum conformal Einstein field equations} are given by
the conditions
\begin{equation}
 {\Sigma}{}_a{}^\bmc{}_\bmb\bme_\bmc=0, \qquad {\Xi}{}^\bmc{}_{\bmd\bma\bmb}=0,
 \qquad  {\Delta}{}_{\bmc\bmd\bmb}=0, \qquad
 \Lambda{}_{\bmb\bmc\bmd}=0. \label{ecfe5}
\end{equation}
In the above equations the fields $\Xi$ and $d_\bma$ ---cfr.
\eqref{Definition:CovectorD}--- are regarded as \emph{conformal gauge
  fields} which are determined by supplementary conditions. In the
present article these gauge conditions will be determined through
conformal geodesics ---see Subsection
\ref{Subsection:ConformalGeodesics} below. In order to account for
this it is convenient to define
\begin{subequations}
\begin{eqnarray}
&& \delta_\bma  \equiv d_\bma -\Xi f_\bma -\hat{\nabla}_\bma\Xi, \label{Supplementary1}\\
&& \gamma_{\bma\bmb} \equiv \hat{L}_{\bma\bmb}
   -\hat{\nabla}_\bma(\Xi^{-1} d_\bmb) -\frac{1}{2}\Xi^{-1}
   S_{\bma\bmb}{}^{\bmc\bmd}d_\bmc d_\bmd +
   \frac{1}{6}\lambda\Xi^{-2}\eta_{\bma\bmb}, \label{Supplementary2}\\
&& \varsigma_{\bma\bmb} \equiv \hat{L}_{[\bma\bmb]}
   -\hat{\nabla}_{[\bma} f_{\bmb]}. \label{Supplementary3}
\end{eqnarray}
\end{subequations}
The conditions
\begin{equation}
\delta_\bma =0, \qquad \gamma_{\bma\bmb}=0, \qquad
\varsigma_{\bma\bmb}=0,
\label{XCFESupplementary}
\end{equation}
will be called the \emph{supplementary conditions}. They play a role
in relating the Einstein field equations to the extended conformal
Einstein field equations and also in the propagation of the constraints.

\medskip
The correspondence between the Einstein field equations and the
extended conformal Einstein field equations is given by the following
---see Proposition 8.3 in \cite{CFEBook}:

\begin{lemma}
  \label{Lemma:XCFEtoEFE}
Let  
\[
(\bme_\bma, \, \hat{\Gamma}_\bma{}^\bmb{}_\bmc,
\hat{L}_{\bma\bmb},\,d^\bma{}_{\bmb\bmc\bmd})
\]
 denote a solution to
the extended conformal Einstein field equations \eqref{ecfe5} for some
choice of the conformal gauge fields $(\Xi,\,d_\bma)$ satisfying the
supplementary conditions \eqref{XCFESupplementary}. Furthermore,
suppose that
\[
 \Xi\neq 0, \qquad \det
(\eta^{\bma\bmb}\bme_\bma\otimes\bme_\bmb)\neq0
\]
 on an open subset
$\mathcal{U}$. Then the metric
\[
\tilde{\bmg}= \Xi^{-2} \eta_{\bma\bmb} \bmomega^\bma\otimes\bmomega^\bmb
\]
is a solution to the Einstein field equations \eqref{EFE} on
$\mathcal{U}$. 
\end{lemma}

\subsubsection{The conformal constraint equations}
\label{Subsection:ConformalConstraints}
The analysis in this article will make use of the \emph{conformal
  constraint Einstein equations}
---i.e. the intrinsic equations implied by the (standard) vacuum conformal
Einstein field equations on a spacelike hypersurface. A derivation of
these equations in its frame form can be found in \cite{CFEBook},
Section 11.4. 

\medskip
Let $\mathcal{S}$ denote a spacelike hypersurface in an unphysical
spacetime $(\mathcal{M},\bmg)$. In the following let $\{ \bme_\bma \}$ denote a $\bmg$-orthonormal
frame adapted to $\mathcal{S}$. That is, the vector $\bme_\bmzero$ is
chosen to coincide with the unit normal vector to the hypersurface and
while the spatial vectors $\{ \bme_\bmi \}$, $\bmi=\bmone, \,
\bmtwo,\, \bmthree$ are intrinsic to $\mathcal{S}$. In our signature
conventions we have that $\bmg(\bme_\bmzero,\bme_\bmzero)=-1$. The extrinsic
curvature is described by the components $\chi_{\bmi\bmj}$ of the
Weingarten tensor. One has that $\chi_{\bmi\bmj}=\chi_{\bmj\bmi}$ and,
moreover
\[
\chi_{\bmi\bmj} = -\Gamma_\bmi{}^\bmzero{}_\bmj.
\]
We denote by $\Omega$ the restriction of the spacetime conformal
factor $\Xi$ to $\mathcal{S}$ and by $\Sigma$ the normal component of
the gradient of $\Xi$. The field $l_{\bmi\bmj}$ denotes the components of the
Schouten tensor of the induced metric $h_{ij}$ on $\mathcal{S}$. 

\smallskip
With the above conventions, the conformal constraint equations in the vacuum case are given by
---see \cite{CFEBook}:
\begin{subequations}
\begin{eqnarray}
&& \label{co1}  D_\bmi D_\bmj \Omega =  \Sigma \chi_{\bmi\bmj} - \Omega L_{\bmi\bmj} + s h_{\bmi\bmj}, \\
&& \label{co2} D_\bmi \Sigma= {\chi_\bmi}^\bmk D_\bmk \Omega - \Omega L_\bmi, \\
&& \label{co3} D_\bmi s=  L_\bmi \Sigma - L_{\bmi\bmk} D^\bmk \Omega, \\
&&\label{co4} D_\bmi L_{\bmj\bmk} - D_\bmj L_{\bmi\bmk}=  \Sigma d_{\bmk\bmi\bmj} + D^\bml \Omega d_{\bml\bmk\bmi\bmj} - (\chi_{\bmi\bmk} L_\bmj - \chi_{\bmj\bmk} L_\bmi), \\
&&\label{co5} D_\bmi L_\bmj - D_\bmj L_\bmi=  D^\bml \Omega d_{\bml\bmi\bmj} +{\chi_\bmi}^\bmk L_{\bmj\bmk} - {\chi_\bmj}^\bmk L_{\bmi\bmk}, \\
&&\label{co6} D^\bmk d_{\bmk\bmi\bmj}= - ({\chi^\bmk}_\bmi d_{\bmj\bmk} - {\chi^\bmk}_\bmj d_{\bmi\bmk}), \\
&&\label{co7} D^\bmi d_{\bmi\bmj}=\chi^{\bmi\bmk} d_{\bmi\bmj\bmk}, \\
&&\label{co8} \lambda= 6 \Omega s + 3 \Sigma^2 - 3 D_\bmk \Omega D^\bmk \Omega, \\
&&\label{co9} D_\bmj \chi_{\bmk\bmi} - D_\bmk \chi_{\bmj\bmi} =\Omega d_{\bmi\bmj\bmk} + h_{\bmi\bmj} L_\bmk - h_{\bmi\bmk} L_\bmj, \\
&&\label{co10}l_{\bmi\bmj}= \Omega d_{\bmi\bmj} + L_{\bmi\bmj} - \chi (
 \chi_{\bmi\bmj} - \frac{1}{4} \chi h_{\bmi\bmj} ) + \chi_{\bmk\bmi}{\chi_\bmj}^\bmk -
\frac{1}{4}\chi_{\bmk\bml}\chi^{\bmk\bml} h_{\bmi\bmj},
\end{eqnarray}
\end{subequations}
with the understanding that 
\[
h_{\bmi\bmj}\equiv g_{\bmi\bmj}=\delta_{\bmi\bmj}
\]
 and where we have defined
\[
L_\bmi\equiv L_{\bmzero\bmi}, \qquad d_{\bmi\bmj}\equiv
d_{\bmzero\bmi\bmzero\bmj}, \qquad d_{\bmi\bmj\bmk}\equiv d_{\bmi\bmzero\bmj\bmk}.
\]
The fields $d_{\bmi\bmj}$ and $d_{\bmi\bmj\bmk}$ correspond,
respectively, to the electric and magnetic parts of the rescaled Weyl
tensor. The scalar $s$ denotes the \emph{Friedrich scalar} defined as
\[
s \equiv \frac{1}{4}\nabla_a\nabla^a \Xi + \frac{1}{24}R \Xi,
\]
with $R$ the Ricci scalar of the metric $\bmg$. Finally,
$L_{\bmi\bmj}$ denote the spatial components of the Schouten tensor of
$\bmg$.

\subsection{Conformal geodesics}
\label{Subsection:ConformalGeodesics}
The gauge to be used to analyse the dynamics of perturbations of the
Schwarzschild-de Sitter spacetime is based on certain conformally
invariant objects known as \emph{conformal geodesics}. Conformal
geodesics allow the use of \emph{conformal Gaussian systems} in which
a certain canonical conformal factor gives an \emph{a priori} (coordinate) location
of the conformal boundary. This is in contrast with other conformal
gauges in which the conformal factor is an unknown. 

\subsubsection{Basic definitions}
A \textbf{\em conformal geodesic}\index{conformal geodesic!definition} on a spacetime
$(\tilde{\mathcal{M}},\tilde{\bmg})$ is a pair
  $(x(\tau),\bmbeta(\tau))$ consisting of a curve $x(\tau)$ on
  $\tilde{\mathcal{M}}$, $\tau \in I \subset \mathbb{R}$, with tangent
  $\dot{\bmx}(\tau)$ and a covector $\bmbeta(\tau)$ along $x(\tau)$
  satisfying the equations
\begin{subequations}
\begin{eqnarray}
&& \tilde{\nabla}_{\dot{\bmx}} \dot{\bmx} = -2 \langle \bmbeta,
\dot{\bmx} \rangle
\dot{\bmx} + \tilde{\bmg}(\dot{\bmx},\dot{\bmx}) \bmbeta^\sharp, \label{ConformalGeodesicEquation1}\\
&& \tilde{\nabla}_{\dot{\bmx}} \bmbeta = \langle \bmbeta, \dot{\bmx}
\rangle \bmbeta - \displaystyle\frac{1}{2} \tilde{\bmg}^\sharp (\bmbeta,\bmbeta)
\dot{\bmx}^\flat + \tilde{\bmL}(\dot{\bmx},\cdot), \label{ConformalGeodesicEquation2}
\end{eqnarray}
\end{subequations}
where $\tilde{\bmL}$ denotes the \emph{Schouten tensor} of the
Levi-Civita connection $\tilde{\bmnabla}$. A vector $\bmv$ is said to
be \emph{Weyl propagated} if along $x(\tau)$ it satisfies the equation
\begin{equation}
\tilde{\nabla}_{\dot{\bmx}} \bmv = -\langle \bmbeta,\bmv\rangle
\dot{\bmx} -\langle \bmbeta,\dot{\bmx}\rangle \bmv +
\tilde{\bmg}(\bmv,\dot{\bmx}) \bmbeta^\sharp.
\label{WeylPropagation}
\end{equation}

\subsubsection{The conformal factor associated to a congruence of
  conformal geodesics}
A congruence of conformal geodesics can be used to single out a metric
$\bmg\in[\tilde{\bmg}]$ by means of a conformal factor $\Theta$ such
that 
\begin{equation}
\bmg(\dot{\bmx},\dot{\bmx})=-1, \qquad \bmg=\Theta^2\tilde{\bmg}. 
\label{CG:NormalisationCondition}
\end{equation}
From the above conditions, it follows that 
\[
\dot{\Theta} = \langle \bmbeta, \dot{\bmx}\rangle \Theta. 
\]
Taking further derivatives with respect to $\tau$ and using the
conformal geodesic equations
\eqref{ConformalGeodesicEquation1}-\eqref{ConformalGeodesicEquation2}
together with the Einstein field equations \eqref{EFE} leads to the
relation
\[
\dddot{\Theta}=0.
\]
From the latter it follows the following result:

\begin{lemma}
  \label{Lemma:ConformalFactor}
Let $(\tilde{\mathcal{M}},\tilde{\bmg})$ denote an Einstein
spacetime. Suppose that $(x(\tau),\bmbeta(\tau))$ is a solution to the
conformal geodesic equations
\eqref{ConformalGeodesicEquation1}-\eqref{ConformalGeodesicEquation2}
and that $\{ \bme_\bma \}$ is a $\bmg$-orthonormal frame propagated
along the curve according to equation \eqref{WeylPropagation}. If
$\Theta$ satisfies \eqref{CG:NormalisationCondition}, then one has
that
\begin{equation}
\Theta(\tau) = \Theta_\star + \dot{\Theta}_\star (\tau-\tau_\star)
+\frac{1}{2}\ddot{\Theta}_\star (\tau-\tau_\star)^2,
\label{CanonicalConformalFactorTheta}
\end{equation}
where the coefficients
\[
\Theta_\star\equiv \Theta(\tau_\star), \qquad 
\dot{\Theta}_\star\equiv \dot{\Theta}(\tau_\star)\qquad
\ddot{\Theta}_\star \equiv \ddot{\Theta}_\star(\tau_\star)
\]
are constant along the conformal geodesic and are subject to the
constraints
\[
\dot{\Theta}_\star = \langle \bmbeta_\star,\dot{\bmx}_\star\rangle
\Theta_\star, \qquad \Theta_\star \ddot{\Theta}_\star =
\frac{1}{2}\tilde{\bmg}^\sharp (\bmbeta_\star,\bmbeta_\star) + \frac{1}{6}\lambda.
\]
Moreover, along each conformal geodesic one has that
\[
\Theta \beta_\bmzero =\dot{\Theta}, \qquad \Theta\beta_\bmi
=\Theta_\star \beta_{\bmi \star},
\]
where $\beta_\bma \equiv \langle \bmbeta,\bme_\bma\rangle$. 
  \end{lemma}

A proof of the above result can be found in \cite{CFEBook},
Proposition 5.1 in Section 5.5.5.

\begin{remark}
{\em Thus, if a spacetime can be covered by a non-intersecting congruence
of conformal geodesics, then the location of the conformal boundary is
known \emph{a priori} in terms of data at a fiduciary initial
hypersurface $\mathcal{S}_\star$. }
\end{remark}

\subsubsection{The $\tilde{g}$-adapted conformal geodesic equations}
As a consequence of the normalisation condition
\eqref{CG:NormalisationCondition}, the parameter $\tau$ is the
$\bmg$-proper time of the curve $x(\tau)$. In some computations it is
more convenient to consider a parametrisation in terms of a
$\tilde{\bmg}$-proper time $\tilde{\tau}$ as it allows to work
directly with 
the physical (i.e. non-conformally rescaled) metric.  To this end, consider the
parameter transformation $\tilde{\tau}=\tilde{\tau}(\tau)$ given by
\begin{equation}
\frac{\mbox{d}\tau}{\mbox{d}\tilde{\tau}} = \Theta,\qquad \mbox{so that}
\qquad \tilde{\tau} = \tilde{\tau}_\star + \int_{\tau_\star}^{\tau}
\frac{\mbox{ds}
  }{\Theta(\mbox{s})},
\label{CG:ChangeOfParameterFormula}
\end{equation}
with inverse $\tau=\tau(\tilde{\tau})$. In what follows, write
$\tilde{x}(\tilde{\tau})\equiv {x}(\tau(\tilde{\tau}))$. It can then
be verified
that
\begin{equation}
\tilde{\bmx}'
\equiv  \frac{\mbox{d}\tilde{x}}{\mbox{d}\tilde{\tau}} = \frac{\mbox{d}\tau}{\mbox{d}\tilde{\tau}} \frac{\mbox{d}x}{\mbox{d}\tau}= \Theta
\dot{\bm x},
\label{CG:OneFormSplitAncillary1}
\end{equation}
so that $\tilde{\bmg}(\tilde{\bmx}',\tilde{\bmx}')=-1$. Hence, $\tilde{\tau}$
is, indeed, the $\tilde{\bmg}$-proper time of the curve $\tilde{x}$.
Now, consider the split
\begin{equation}
\bmbeta=\tilde{\bmbeta} + \varpi \dot{\bmx}^\flat, \qquad \varpi \equiv 
\frac{\langle \bmbeta,
  \dot{\bmx}\rangle}{\tilde{\bmg}(\dot{\bmx},\dot{\bmx})},
\label{CG:OneFormSplit}
\end{equation}
where the covector $\tilde{\bmbeta}$ satisfies
\begin{equation}
\langle \tilde{\bmbeta}, \dot{\bmx}\rangle =0, \qquad {\bm
  g}^\sharp({\bmbeta},{\bmbeta})
=\langle {\bmbeta},\dot{\bmx}\rangle^2 + {\bmg}^\sharp(\tilde{\bmbeta},\tilde{\bmbeta}).
\label{CG:OneFormSplitAncillary2}
\end{equation}
 It can be readily verified that
\begin{equation}
\tilde{\bmg}(\dot{\bmx},\dot{\bmx})=-\Theta^{-2}, \qquad \langle
\bmbeta,\dot{\bmx}\rangle= \Theta^{-1}\dot{\Theta}, \qquad \varpi = \Theta \dot{\Theta}.
\label{CG:OneFormSplitAncillary3}
\end{equation}
Using the split \eqref{CG:OneFormSplit} in equations
\eqref{ConformalGeodesicEquation1}-\eqref{ConformalGeodesicEquation2}
and taking into account the relations in
\eqref{CG:OneFormSplitAncillary1}, \eqref{CG:OneFormSplitAncillary2}
and \eqref{CG:OneFormSplitAncillary3}  
one obtains the following $\tilde{\bm g}$-\textbf{\em adapted
  equations for the conformal geodesics}: 
\index{conformal geodesic!$\tilde{\bmg}$-adapted equations}
\begin{subequations}
\begin{eqnarray}
&& \tilde{\nabla}_{\tilde{\bmx}'} \tilde{\bmx}' = \tilde{\bmbeta}^\sharp, \label{PhysicalMetricAdaptedCC1}\\
&& \tilde{\nabla}_{\tilde{\bmx}'} \tilde{\bmbeta} = \tilde\beta^2 \tilde{\bmx}'{}^\flat + \tilde{\bmL}(\tilde{\bmx}',\cdot) -\tilde{\bmL}(\tilde{\bmx}',\tilde{\bmx}')\tilde{\bmx}^{\prime\flat}, \label{PhysicalMetricAdaptedCC2}
\end{eqnarray}
\end{subequations}
with $\tilde\beta^2\equiv \tilde{\bm
  g}^\sharp(\tilde{\bmbeta},\tilde{\bmbeta})$ ---observe that as a
consequence of \eqref{CG:OneFormSplitAncillary2} the covector
$\tilde{\bmbeta}$ is spacelike and, thus, the definition of $\tilde\beta^2$
makes sense. For an Einstein space one has that
\[
\tilde{\bmL}=\frac{1}{6}\lambda\tilde{\bmg}.
\]

The Weyl propagation equation
\eqref{WeylPropagation} can also be cast in a
$\tilde{\bmg}$-adapted form. A calculation shows that
\begin{equation}
\hat{\nabla}_{\tilde{x}'} (\Theta\bmv) = -\langle \tilde{\bmbeta},
\Theta\bmv \rangle \tilde{\bmx}' +\tilde{\bmg}
(\Theta\bmv,\tilde{\bmx}'\rangle \tilde{\bmbeta}^\sharp.
\label{WeylPropagationGTildaAdapted}
\end{equation}

\subsubsection{Conformal
  Gaussian gauges}
Now, consider a region $\mathcal{U}$ of the spacetime
$(\tilde{\mathcal{M}},\tilde{\bmg})$ covered by a non-intersecting
congruence of conformal geodesics
$({\bmx}(\tau),\bmbeta(\tau))$. From Lemma
   \ref{Lemma:ConformalFactor} follows that the requirement
   $\bmg(\dot{\bmx},\dot{\bmx})=-1$ singles out a \emph{canonical
   representative} $\bmg$ of the conformal class $[\tilde{\bmg}]$ with
 an explicitly known conformal factor as given by the formula
 \eqref{CanonicalConformalFactorTheta}.

 Now, let $\{ \bme_\bma \}$ denote a $\bmg$-orthonormal frame which is
 Weyl propagated along the conformal geodesics. It is natural to set
 $\bme_\bmzero=\dot{\bmx}$. To every congruence of conformal geodesics
 one can associate a Weyl connection $\hat{\nabla}_a$ by setting $f_a=\beta_a$. It follows that
 for this connection one has
 \[
 \hat{\Gamma}_\bmzero{}^\bma{}_\bmb=0, \qquad f_\bmzero=0, \qquad \hat{L}_{\bmzero \bma}=0.
 \]
 This gauge choice can be supplemented by choosing the parameter $\tau$
 of the conformal geodesics as the time coordinate so that
 \[
 \bme_\bmzero= \bmpartial_\tau. 
\]
In the following, it will be assumed that initial data for the
congruence of conformal geodesics is prescribed on a fiduciary spacelike
hypersurface $\mathcal{S}_\star$. On $\mathcal{S}_\star$ one can
choose some local coordinates $\underline{x}=(x^\alpha)$. If the
congruence is non-intersecting, one can extend the coordinates
$\underline{x}$ off  $\mathcal{S}_\star$ by requiring them to remain
constant along the conformal geodesic which intersects
$\mathcal{S}_\star$ at the point $p$ on $\mathcal{S}_\star$ with
coordinates $\underline{x}$. The spacetime coordinates
$\overline{x}=(\tau,x^\alpha)$ obtained in this way are known as
\emph{conformal Gaussian coordinates}. More generally, the collection
of conformal factor $\Theta$, Weyl propagated frame $\{ \bme_\bma \}$
and coordinates $(\tau,x^\alpha)$ obtained by the procedure outlined
in the previous paragraph is known as a \emph{conformal Gaussian gauge
  system}. More details on this construction can be found in
\cite{CFEBook}, Section 13.4.1.

\section{The Schwarzschild-de Sitter spacetime}
\label{Section:BackgroundSolution}

The purpose of this section is to discuss the key properties of the
Schwarzschild-de Sitter spacetime that will be used in our argument on
the stability of the Cosmological region of this exact solution.

\subsection{Basic properties}
The \emph{Schwarzschild-de Sitter spacetime}, $(\tilde{\mathcal{M}},\mathring{\tilde{\bmg}})$, is the solution to the 
  vacuum Einstein field equations with positive Cosmological constant
\begin{equation}
\tilde{R}_{ab} =\lambda \tilde{g}_{ab}, \qquad \lambda>0
\label{EFE}
\end{equation}
with $\tilde{\mathcal{M}}=\mathbb{R}\times\mathbb{R}^+\times
\mathbb{S}^2$ and line element given in \emph{standard coordinates}
$(t,r,\theta,\varphi)$ by
\begin{equation}	
\mathring{\tilde{\bmg}} =-\bigg{(} 1-\frac{2m}{r}-\frac{\lambda}{3}r^2\bigg{)}\mathbf{d} t  \otimes \mathbf{d} t +
\bigg{(}1-\frac{2m}{r}-\frac{\lambda}{3}r^2\bigg{)}^{-1} \mathbf{d} r \otimes  \mathbf{d} r + r^2  \bm{\sigma}
\end{equation}
where
\[
\bmsigma\equiv \mathbf{d}\theta\otimes\mathbf{d}\theta +\sin^2\theta \mathbf{d}\varphi\otimes\mathbf{d}\varphi,
\]
denotes the standard metric on $\mathbb{S}^2$. The coordinates
$(t,r,\theta,\varphi)$ take the range
\[
t\in(-\infty,\infty), \qquad r\in (0,\infty), \qquad
\theta\in (0,\pi), \qquad \varphi\in[0,2\pi).
\]
This line element can be rescaled so to that 
\begin{equation}	
\mathring{\tilde{\bmg}} =-D(r)\mathbf{d} t  \otimes \mathbf{d} t +
\frac{1}{D(r)} \mathbf{d} r \otimes  \mathbf{d} r + r^2  \bm{\sigma},
\label{BackgroundPhysicalMetric}
\end{equation}
where
\[
  M \equiv 2m \sqrt{\frac{ \lambda}{3}}
\]
and
\[
D(r)\equiv 1 - \frac{M}{r} - r^2.
\]
In our conventions $M$, $r$ and $\lambda$ are dimensionless quantities.

\subsection{Horizons and global structure}

The location of the horizons of the Schwarzschild-de Sitter spacetime
follows from the analysis of the zeros of the function $D(r)$ in the
line element \eqref{BackgroundPhysicalMetric}. 

\medskip
Since $\lambda>0$, then the function $D(r)$ can be factorised as
\[
D(r) = -\frac{1}{r}(r-r_b)(r-r_c)(r-r_-), 
\]
where $r_b$ and $r_c$ are, in general, distinct positive roots of $D(r)$ and $r_-$ is
a negative root. Moreover, one has that 
\[
0<r_b<r_c, \qquad r_c+r_b+r_-=0.
\]
The root $r_b$ corresponds to a black hole-type of horizon and $r_c$
to a Cosmological de Sitter-like type of horizon. One can verify that 
\begin{eqnarray*}
&& D(r)>0 \qquad \mbox{for} \qquad r_b<r<r_c, \\
&& D(r)<0 \qquad \mbox{for} \qquad 0<r<r_b \qquad \mbox{and} \qquad r>r_c.
\end{eqnarray*}
Accordingly, $\mathring{\tilde{\bmg}}$ is static in the region
$r_b<r<r_c$ between the horizons. There are no other static regions
outside this range. 

\medskip
Using Cardano's formula for cubic equations, we have
\begin{subequations}
\begin{eqnarray}
&& r_-=- \frac{2}{\sqrt{3}} \cos \bigg{(}\frac{\phi}{3} \bigg{)}, \label{Horizon1}\\
&& r_b= \frac{1}{\sqrt{3}} \bigg{(} \cos \bigg{(}\frac{\phi}{3} \bigg{)}- \sqrt{3}  \sin \bigg{(}\frac{\phi}{3} \bigg{)} \bigg {)}, \label{Horizon2}\\
&& r_c= \frac{1}{\sqrt{3}} \bigg{(} \cos \bigg{(}\frac{\phi}{3}
   \bigg{)}+ \sqrt{3}  \sin \bigg{(}\frac{\phi}{3} \bigg{)} \bigg
   {)}. \label{Horizon3}
\end{eqnarray}
\end{subequations}
where the parameter $\phi$ is defined through the relation
\begin{equation}
  M= \frac{2 \cos \phi}{3 \sqrt{3}}, \qquad \phi \in \bigg{(} 0,
  \frac{\pi}{2} \bigg{)}.
  \label{DefinitionM}
\end{equation}
In the sub-extremal case we have that $0 < M < 2/3 \sqrt{3}$
and $\phi \in (0, \pi/2)$. This describes a black hole in a
Cosmological setting. The extremal case corresponds to the value
$\phi=0$ for which $M=2/3\sqrt{3}$ ---in this case the
Cosmological and black hole horizons coincide. Finally, the hyper-extremal
case is characterised by the condition $M>2/3\sqrt{3}$ ---in this case
the spacetime contains no horizons.

\medskip
The Penrose diagram of the Schwarzschild-de Sitter is well
known ---see Figure \ref{Figure:SdSPenroseDiagram}. Details of its
construction can be found in e.g. \cite{GriPod09,CFEBook}. 

\begin{figure}[t]
\centering
\includegraphics[width=1\textwidth]{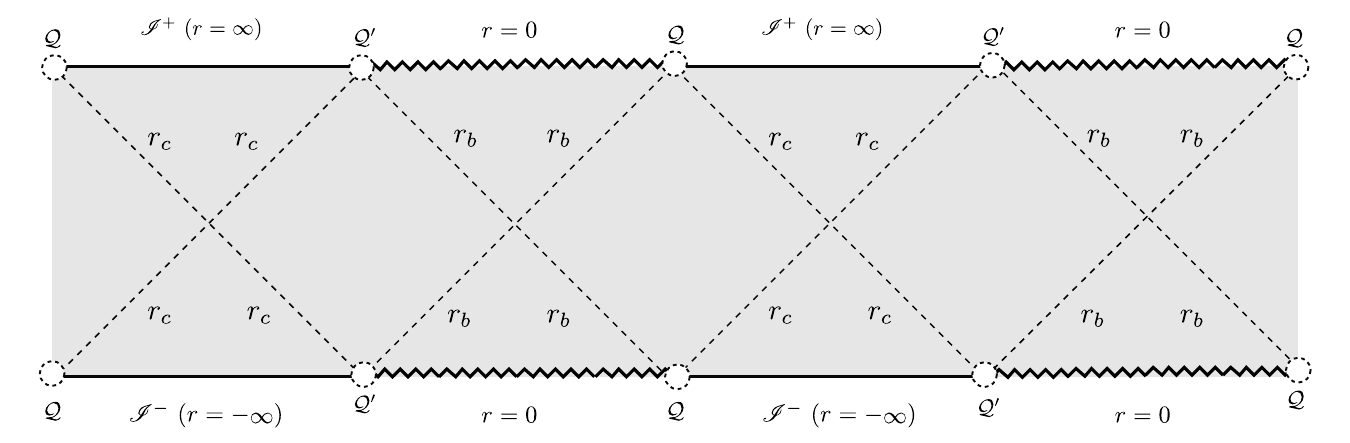}
\caption{Penrose diagram of the sub-extremal Schwarzschild-de Sitter
  spacetime. The serrated line denotes the location of the
  singularity; the continuous black line denotes the conformal
  boundary; the dashed line shows the location of the black hole and
  Cosmological horizons denoted by $\mathcal{H}_{b}$ and
  $\mathcal{H}_{c}$
  respectively. As described in the main text, these horizons are located
at $r=r_{b}$ and $r=r_{c}$.  The excluded points $\mathcal{Q}$
  and $\mathcal{Q'}$ where the singularity seems to meet the
  conformal boundary correspond to asymptotic regions of the
  spacetime that does not belong to the singularity nor the
  conformal boundary.  }
\label{Figure:SdSPenroseDiagram}
\end{figure}

\subsection{Other coordinate systems}
In our analysis, we will also make use of \emph{retarded and advanced
  Eddington-Finkelstein null coordinates} defined by
\begin{equation}\label{EFcoords}
u\equiv t-r^*, \qquad v\equiv t+r^*,
\end{equation}
where $r^*$ is the \emph{tortoise coordinate} given by
\begin{equation}\label{Tortoise}
r^*(t) \equiv \int \frac{\mathrm{d}r}{D(r)}, \qquad
\lim_{r\rightarrow\infty} r^*(r)=0. 
\end{equation}
It follows that $u,\,v\in \mathbb{R}$. In terms of these coordinates
the metric $\mathring{\tilde{\bmg}}$ takes, respectively, the forms
\begin{eqnarray*}
&& \mathring{\tilde{\bmg}} =- D(r) \mathbf{d}u\otimes\mathbf{d}u
+(\mathbf{d}u\otimes\mathbf{d}r +\mathbf{d}r\otimes\mathbf{d}u)
   +r^2\bmsigma, \\
  && \mathring{\tilde{\bmg}} =- D(r) \mathbf{d}v\otimes\mathbf{d}v
+(\mathbf{d}v\otimes\mathbf{d}r +\mathbf{d}r\otimes\mathbf{d}v)
+r^2\bmsigma. 
\end{eqnarray*}

\medskip
In order to compute the Penrose diagrams, Figures \ref{Figure:HypersurfaceswithconstantronPenroseDiagramsSdS} and \ref{Figure:PenroseDiagramConformalGeodesics}, we make use of \emph{Kruskal
  coordinates} defined via
\[
  U\equiv\frac{1}{2}\exp(bu), \qquad V\equiv\frac{1}{2}\exp(b v)
\]
where $u$ and $v$ are the Eddington-Finkelstein coordinates as defined
in \eqref{EFcoords} and $b$ is a constant which can be freely
chosen. A further change of coordinates is provided by 
\[
  T\equiv U + V, \qquad \Psi\equiv U-V.
 \]
These coordinates are related to $r$ and $t$ via
\[
  T(r,t)= \cosh(b t) \exp(b r^*(r)), \qquad \Psi(r,t)= \sinh(b t)
  \exp(b r^*(r)).
  \]
Then by recalling that
\[
  r_- <0<r_b<r_c \qquad {\rm and} \qquad  r_-+r_b+r_c=0,
\]
the equation of $r^*(r)$ as defined by \eqref{Tortoise} renders
 \[
   r^*(r)=-\frac{ r_b\ln(r-r_b)}{(r_b -r_c)(2r_b + r_c)} +\frac{
     r_c\ln(r-r_c)}{r_b^2 + r_b r_c -2 r_c^2}+ \frac{
     (r_b+r_c)\ln(r+r_b+r_c)}{(2r_b +r_c)(r_b +2 r_c)}.
 \]
Hence, in order to have coordinates which are regular down to the
Cosmological horizon, the constant $b$ must be given  by 
\[
b = \frac{ r_b^2 + r_br_c - 2r_c^2}{2 r_c}.
\]

\section{Construction of a conformal Gaussian gauge in the
  Cosmological region}
\label{Section:CGSdS}

The hyperbolic reduction of the extended conformal Einstein field
equations to be used in this article makes use of a conformal Gaussian
gauge system ---i.e. coordinates and frame are propagated along a
suitable congruence of conformal geodesics. This congruence provides,
in turn, a canonical representative of the conformal class of a
solution to the Einstein field equations ---see e.g. Proposition 5.1 in
\cite{CFEBook}. 

\medskip
A class of non-intersecting conformal geodesics which cover the whole maximal extension
of the sub-extremal Schwarzschild-de Sitter spacetime has been studied
in \cite{GarGasVal18}. The main outcome of
the analysis in that reference is that the resulting congruence covers the whole
maximal analytic extension of the spacetime and, accordingly, provides
a global system of coordinates ---modulo the usual difficulties with the
prescription of coordinates on $\mathbb{S}^2$. This congruence is prescribed in terms of data
prescribed on a Cauchy hypersurface of the spacetime. In the present article,
we are interested in the evolution of perturbations of the
Schwarzschild-de Sitter spacetime from data prescribed on
hypersurfaces of constant coordinate $r$ in the Cosmological region of
the spacetime. Thus, the congruence of conformal geodesics constructed
in \cite{GarGasVal18} is of no direct use to us. Consequently, in this
section, we study a class of conformal geodesics of the
Schwarzschild-de Sitter spacetime which is prescribed in terms of data
on hypersurfaces of constant $r$ in the Cosmological region. These curves turn out to be geodesics of the
physical metric $\tilde{\bmg}$ and intersect the conformal boundary
orthogonally.

\subsection{Basic setup}
In the following, it is assumed that
\[
r_c < r < \infty
\]
corresponding to the Cosmological region of the Schwarzschild-de
Sitter spacetime. Given a fixed $r=r_\star$ we denote by
$\mathcal{S}_{r_\star}$ (or $\mathcal{S}_\star$ for short) the
spacelike 
hypersurfaces of constant $r=r_\star$ in this region ---see Figure \ref{Figure:HypersurfaceswithconstantronPenroseDiagramsSdS}.  Points on
$\mathcal{S}_\star$ can be described in terms of the coordinates
$(t,\theta,\varphi)$. 

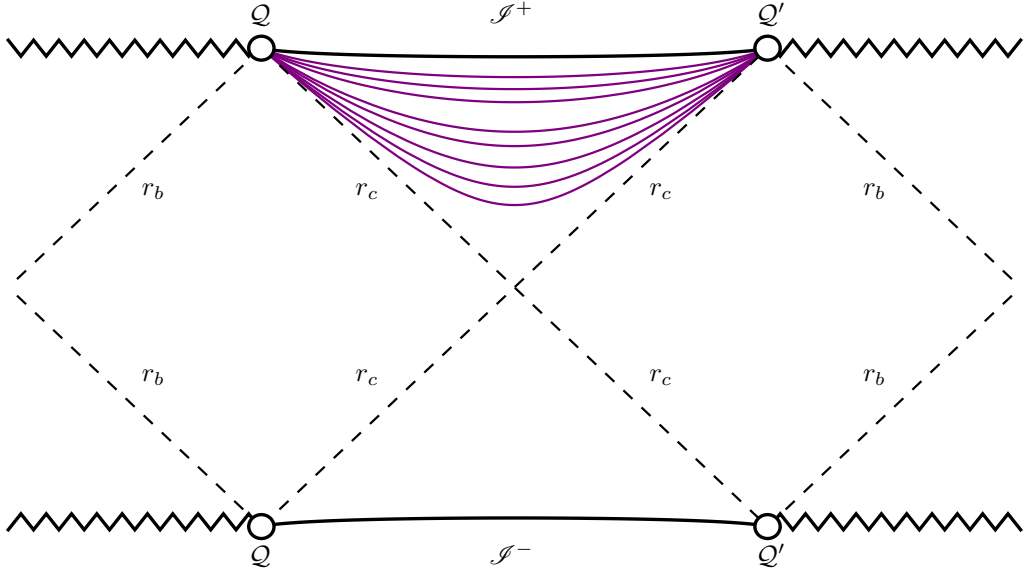
\begin{figure}[t]
\centering
\includegraphics[width=1\textwidth]{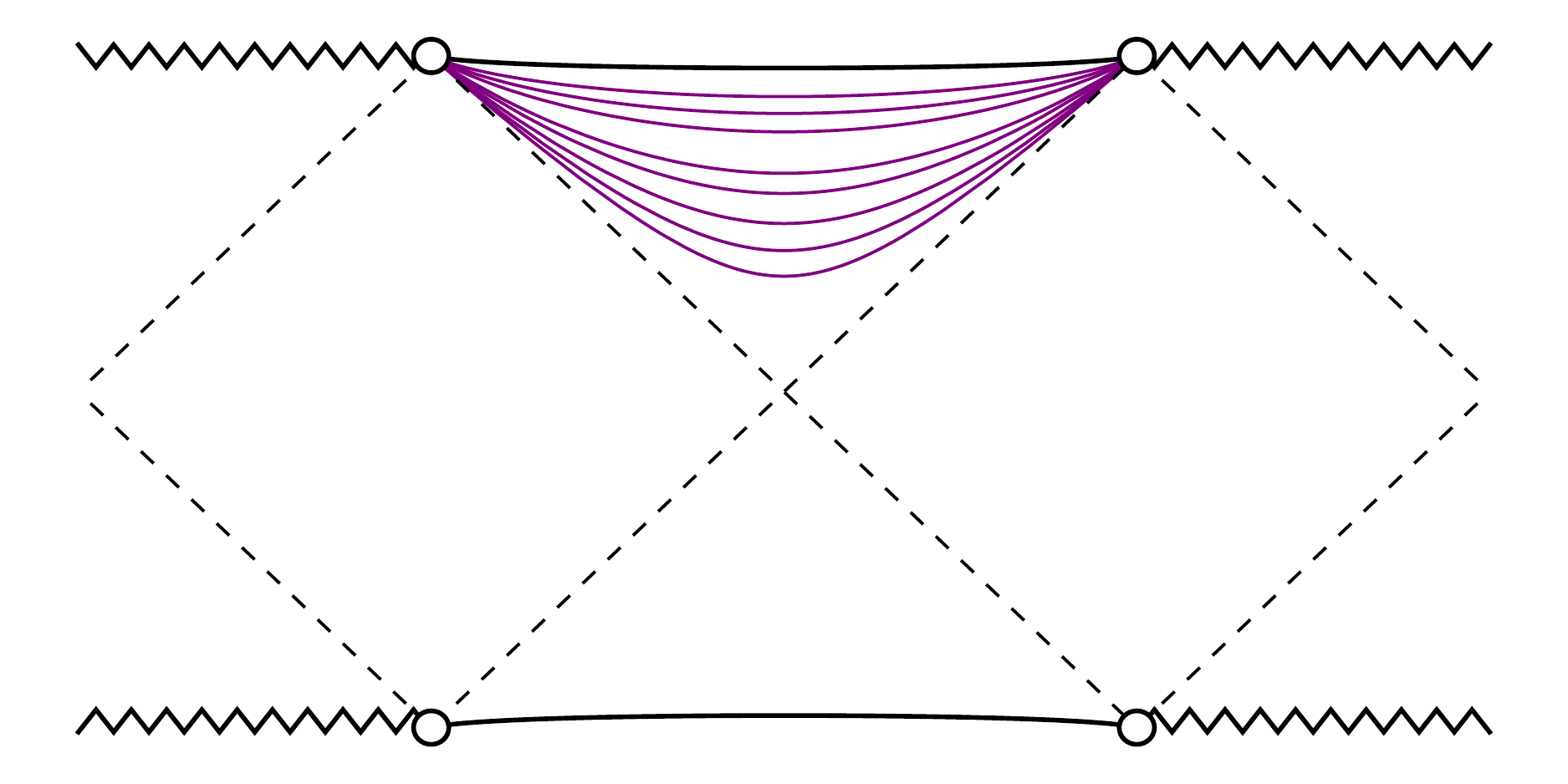}
\put(-310,205){$\mathcal{Q}$}
\put(-350,140){$r_b$}
\put(-270,140){$r_c$}
\put(-220,205){$\mathscr{I}^+$}
\put(-160,140){$r_c$}
\put(-80,140){$r_b$}
\put(-120,205){$\mathcal{Q}'$}
\put(-270,70){$r_c$}
\put(-350,70){$r_b$}
\put(-220,1){$\mathscr{I}^-$}
\put(-160,70){$r_c$}
\put(-80,70){$r_b$}
\put(-310,1){$\mathcal{Q}$}
\put(-120,1){$\mathcal{Q}'$}
\caption{Hypersurfaces with constant $r$ are plotted on the Penrose diagram of the Cosmological region of the sub-extremal Schwarzschild-de Sitter spacetime.  }
\label{Figure:HypersurfaceswithconstantronPenroseDiagramsSdS}
\end{figure}

\subsubsection{Initial data for the congruence}
\label{Subsection:InitialDataCongruence}
In order to prescribe the congruence of conformal geodesics, we follow
the general strategy outlined in \cite{Fri03c,GarGasVal18}. This
requires prescribing the value of a conformal factor $\Theta_\star$
over $\mathcal{S}_\star$. We
will only be interested on prescribing the data on compact subsets of
$\mathcal{S}_\star$ so it is natural to require that
\[
  \Theta_\star=1, \qquad \dot{\Theta}_\star=0.
\]
The second condition implies that the resulting conformal factor will
have a time reflection symmetry with respect to
$\mathcal{S}_\star$. Now, following  \cite{Fri03c,GarGasVal18} we
require that
\[
\tilde{\bmx}'_\star \perp \mathcal{S}_\star, \qquad \tilde{\beta}_\star=\Theta_\star^{-1}{\rm d}\Theta_\star.
\]
The latter, in turn, implies that 
\begin{equation}
\label{InitialData1} t=t_\star \qquad  t{}'{}_\star=
\frac{1}{\sqrt{D_\star}}, \qquad r{}'{}_\star=0, \qquad
\tilde{\beta}_{t \star}=0, \qquad\tilde{\beta}_{r \star}=0, 
\end{equation}
where $t_\star\in (-t_\bullet,t_\bullet)$ for some $t_\bullet\in
\mathbb{R}^+$. Notice that the tangent vector $\tilde{x}'$ coincides with the future unit
normal to $\tilde{\mathcal{S}}$. 

\medskip
Given a sufficiently large constant $t_\bullet$ we define
\[
\mathcal{R}_\bullet =\{ p\in \mathcal{S}_\star \;|\; t(p) \in (-t_\bullet,t_\bullet) \}.
\]
The constant $t_\bullet$ will be assumed to be large enough so that
$D^+(\mathcal{R}_\bullet)\cap \mathscr{I}^+\neq \varnothing$.

\begin{remark}
{\em The starting point of the curves on $\mathcal{S}_\star$ is
  prescribed in terms of the coordinates $(t,\theta,\varphi)=(t_\star,\theta_\star,\varphi_\star)$ The conditions \eqref{InitialData1} gives rise to a congruence of
conformal geodesics which has a trivial behaviour of the angular
coordinates ---that is, it is spherically symmetric. In other words
effectively analysing the curves on a $2$-dimensional manifold
$\tilde{\mathcal{M}}/{\rm SO}(3)$ with quotient metric $\tilde{\bmell}$ given by
\[
  \tilde{\bmell}= -D(r)\mathbf{d} t \otimes \mathbf{d} t + D^{-1}(r) \mathbf{d} r
  \otimes \mathbf{d} r
  \]
Accordingly, the
only non-trivial parameter characterising each curve of the congruence
is $t_\star$.}
  \end{remark}

\subsubsection{The geodesic equations}
It follows that for the initial data conditions \eqref{InitialData1} one has 
$\beta^2=0$ so that the resulting congruence of conformal geodesics
is, after reparametrisation, a congruence of metric geodesics. This
last observation simplifies the subsequent discussion. The geodesic equations then imply that 
\begin{equation}
 r'=\sqrt{\gamma^2 - D(r)}, \qquad  D(r)t'{}^2-
 \frac{1}{D(r)}r'^2=1, \label{GeodesicEquations}
\end{equation}
where $\gamma$ is a constant. Evaluating at $\mathcal{S}_\star$ one
readily finds that
\[
t_\star'= \frac{|\gamma|}{|D_\star|}.
\]
Observe that since we are in the Cosmological region of the spacetime
we have that $D_\star <0$. Moreover, the unit normal to $\mathcal{S}_\star$ is given by
\[
  \bmn= \bigg{(} \frac{1}{\sqrt{|D_\star|}} \bigg{)} \mathbf{d}r
  \]
while 
\[
  \tilde{\bmx}'{}_\star=\tilde{r}'{}_\star \bmpartial_r + {t}'{}_\star\bmpartial_t .
\]
So, it follows that $\tilde{\bmx}_\star'$ and $\bmn^\sharp$ are
parallel if and only if $\gamma=0$. 

\subsubsection{The conformal factor}
In the following, in order to obtain simpler expressions we set
$\lambda=3$ and $\tau_\star=0$. It follows then from formula \eqref{CanonicalConformalFactorTheta}
that one gets an explicit expression for the conformal factor. Namely,
one has that
\begin{equation}
  \Theta(\tau)=1-\frac{1}{4}\tau^2.
  \label{CanonicalThetaSdS}
  \end{equation}
The roots of $\Theta(\tau)$ are given by
\[
  \tau_+ \equiv 2, \qquad \tau_-\equiv -2.
\]
In the following, we concentrate on the root $\tau_+$ corresponding to
the location of the future conformal boundary $\mathscr{I}^+$. The
relation between the physical proper time $\tilde{\tau}$ and the unphysical proper time $\tau$ is obtained
from equation \eqref{CG:ChangeOfParameterFormula} so that 
\begin{equation}
  \tilde{\tau}= 2 {\rm arctanh}\bigg{(}\frac{\tau}{2} \bigg{)}, \qquad
  \tau=2 {\rm tanh}\bigg{(}\frac{\tilde{\tau}}{2} \bigg{)}.
\label{TransformationProperTime}
\end{equation}
From these expressions, we deduce that
\[
  \tau\rightarrow \tau_\pm=2, \qquad \mbox{as} \quad
  \tilde{\tau}\rightarrow \infty.
\]
Moreover, the conformal factor $\Theta$ can be rewritten in terms of the $\tilde{\bmg}$-proper time $\tilde{\tau}$ as
\[
  \Theta(\tilde{\tau})={\rm sech}^2 \bigg{(}
  \frac{\tilde{\tau}}{2}\bigg{)}.
\]

\begin{remark}
{\em In \cite{FriSch87} it has been shown that conformal
geodesics in an Einstein space will reach the conformal boundary orthogonally if and only
if they are, up to a reparametrisation standard (metric) geodesics. In
the present case, this property can be directly verified using
equations \eqref{GeodesicEquations}. }
  \end{remark}

\subsection{Qualitative analysis of the behaviour of the curves}
Having, in the previous subsection, set up the initial data for the
congruence of conformal geodesics, in this subsection we analyse the
qualitative behaviour of the curves. In particular, we show that the
curves reach the conformal boundary in a finite amount of (conformal)
proper time. Moreover, we also show that the curves do not intersect
in the future of the initial hypersurface $\mathcal{S}_\star$. 

\subsubsection{Behaviour towards the conformal boundary}
Recalling that
\begin{equation}
  r'=\sqrt{|D(r)|}
  \label{rtildedash}
\end{equation}
and observing that $D(r)<0$, it follows that if $r'{}_\star\neq 0$
then, in fact $r'>0$. Moreover, one can show that $r''{}_\star >0$ and
that $r''{}_\star \neq 0$ for $r \in [r_\star, \infty)$. Thus, the
curves escape to the conformal boundary.

\medskip
Now, we show that the congruence of conformal geodesics reaches the
conformal boundary in an infinite amount of the physical proper
time. In order to see this, we observe that $D(r)<0$, consequently
from equation
\[
 r'= \pm \sqrt{|D(r)|}  
\]
it follows that $r(\tilde{\tau})$ is a monotonic function. Moreover, using equations
\[
D(r)=-\frac{1}{r} (r-r_b)(r-r_-)(r-r_c)
\]
and
\[
t'=\frac{|\gamma + \beta r|}{|D(r)|}=0 
\]
we find that
\[
\tilde{\tau}=\bigintss_{r_\star}^{r}\sqrt{
  \frac{\bar{r}}{(\bar{r}-r_b)(\bar{r}-r_c)(\bar{r}
    -r_-)}} {\rm d} \bar{r}. 
\]
It is possible to rewrite this integral in terms of elliptic functions
---see e.g. \cite{Law89}. More precisely, one has that 
\begin{equation}
 \label{eqtau2} \tilde{\tau}=\frac{2 r_\star}{\alpha^2  \sqrt{r_\star
     (\alpha_+ - \alpha_-)}}\bigg{(} \kappa^2 {\rm
   w}+(\alpha^2-\kappa^2)\Pi[\phi,\alpha^2,\kappa] \bigg {)}, 
\end{equation}
where
$\Pi[\phi,\alpha^2,\kappa]$ is the incomplete elliptic integral of the
third kind and
\[{\rm sn}^2 {\rm w}= \bigg{(} \frac{r_c-
r_-}{r_b-r_-} \bigg{)} \bigg{(} \frac{r-
r_b}{r-r_c} \bigg{)}, \qquad \alpha^2 \equiv \frac{r_b
-r_-}{r_c -r_-}, \]
\[ \kappa^2 \equiv \frac{r_c (r_b -
r_-)}{r_\star(r_c-r_-)}, \qquad \phi \equiv {\rm
arcsin}({\rm sn}{\rm w}), 
\] 
with  $\mathrm{sn}$ denotes the Jacobian elliptic function. 
From the previous expressions and the
general theory of elliptic functions it follows that
$\tilde{\tau}(r,r_\star)$ as defined by Equation \eqref{eqtau2} is
an analytic function of its arguments. Moreover, it can be verified
that 
\[
\tilde{\tau}\rightarrow \infty \qquad \mbox{as} \quad r\rightarrow \infty.
\]
Accordingly, as expected, the curves escape to infinity in an infinite
amount of physical proper time. Using the reparametrisation formulae
\eqref{TransformationProperTime} the latter corresponds to a finite
amount of unphysical proper time.

\subsubsection{Analysis of the behaviour of the conformal deviation equation}
In \cite{Fri03c} (see also \cite{GarGasVal18}) it has been shown that for congruences of conformal
geodesics in spherically symmetric spacetimes the behaviour of the
deviation vector of the congruence can be understood by considering
the evolution of a scalar $\tilde{\omega}$ ---see equation (33) in \cite{GarGasVal18}. If this scalar does not
vanish, then the congruence is non-intersecting. Since in the
present case one has $\beta=0$, it follows that the evolution equation
for $\tilde{\omega}$ takes the form
\[ 
\frac{{\rm d}^2 \tilde{\omega}}{{\rm d} \tilde{\tau}^2}=\bigg{(}1+
\frac{M}{r^3}\bigg{)}\tilde{\omega}, \qquad r\equiv r(\tilde{\tau},
r_\star).
\]
Since in our setting $r\geq r_\star>r_c$, it follows that
\[ 
1 + \frac{M}{r^3} >1, 
\]
from where, in turn, one obtains the inequality
\[ 
\frac{{\rm d}^2 \tilde{\omega}}{{\rm d} \tilde{\tau}^2}> \tilde{\omega}. 
\]
Accordingly, the scalars $\tilde{\omega}$ and $\omega \equiv \Theta \tilde{\omega}$ satisfy the inequalities
\[ 
\tilde{\omega}  \geq \bar{\omega}, \qquad \omega \geq \Theta
\bar{\omega}, 
\]
where $\bar{\omega}$ is the solution of 
\[ 
\frac{{\rm d}^2 \bar{\omega}}{{\rm d} \tilde{\tau}^2}=\bar{\omega},
\qquad \bar{\omega}(0,\rho_\star)= \frac{r_\star}{\rho_\star}, \qquad
\bar{\omega}'(0,\rho_\star)=0. 
\]
The solution to this last differential equation is given by 
\[
\bar{\omega}=(r_\star/\rho_\star){\rm cosh}\tilde{\tau}.
\]
Using equations \eqref{CanonicalThetaSdS} and
\eqref{TransformationProperTime} we get the inequality 
\[
 \omega \geq \bigg{(} 1- \frac{ \tau^2}{4}
 \bigg{)}\frac{r_\star}{\rho_\star}{\rm cosh}\bigg{(}2 {\rm
   arctanh}\bigg{(} \frac{\tau}{2}
 \bigg{)}\bigg{)}=\frac{r_\star}{\rho_\star} \bigg{(} 1+ \frac{
   \tau^2}{4} \bigg{)} >0. 
\]
Consequently, we get the limit
\[
\lim_{\tau \to \pm 2} \omega \geq \frac{2 r_\star}{\rho_\star} >0. 
\]
Hence, we conclude that the geodesics with $r_\star > r_\bullet$ which
go to the conformal boundary $\mathscr{I}^+$ located at $\tau=2$ do
not develop any caustics.

\medskip
The discussion of the previous paragraphs can be summarised in the
following:

\begin{proposition}
\label{Proposition:CGNoCaustics}
The congruence of conformal geodesics given by the initial conditions
\eqref{InitialData1} leaving the initial hypersurface
$\mathcal{S}_\star$ reach the conformal boundary $\mathscr{I}^+$
without developing caustics.
\end{proposition}

The content of this Proposition can be visualised in Figure [$\ref{Figure:PenroseDiagramConformalGeodesics}$].

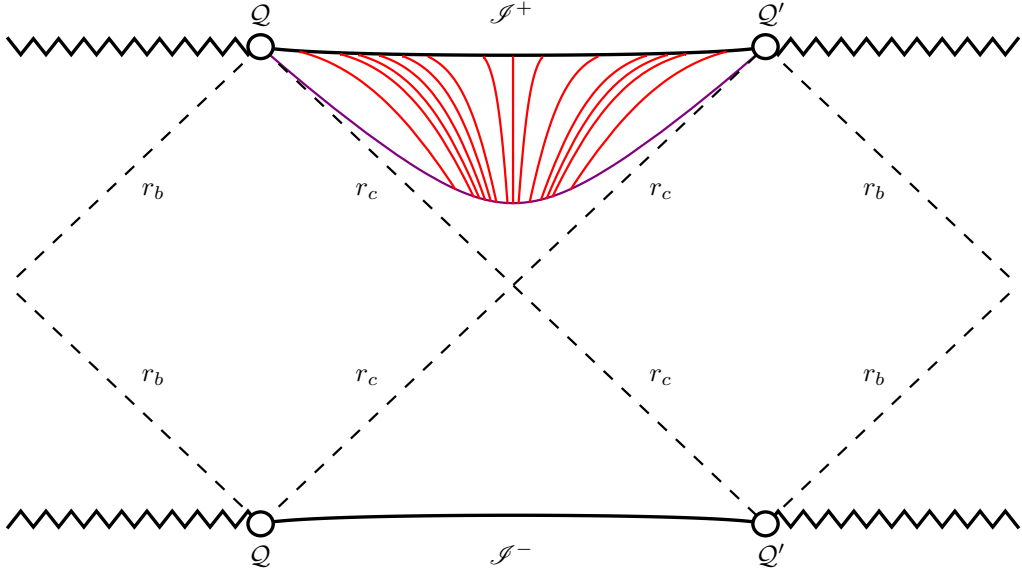
\begin{figure}[t]
\centering
\includegraphics[width=1\textwidth]{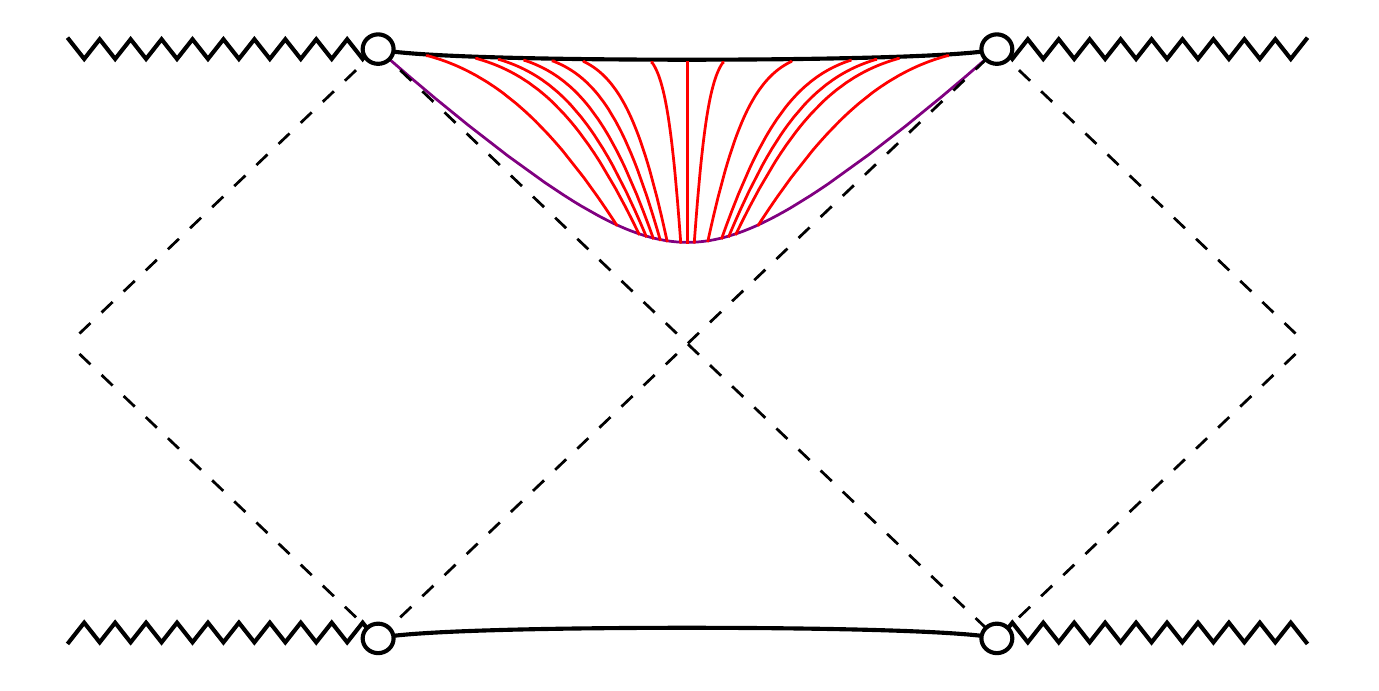}
\put(-310,205){$\mathcal{Q}$}
\put(-350,140){$r_b$}
\put(-270,140){$r_c$}
\put(-220,205){$\mathscr{I}^+$}
\put(-160,140){$r_c$}
\put(-80,140){$r_b$}
\put(-120,205){$\mathcal{Q}'$}
\put(-270,70){$r_c$}
\put(-350,70){$r_b$}
\put(-220,1){$\mathscr{I}^-$}
\put(-160,70){$r_c$}
\put(-80,70){$r_b$}
\put(-310,1){$\mathcal{Q}$}
\put(-120,1){$\mathcal{Q}'$}
\caption{The conformal geodesics are plotted on the Penrose diagram of
the Cosmological region of the sub-extremal Schwarzschild-de
Sitter spacetime. The purple line represents the initial hypersurface $\mathcal{S}_\star$ corresponding to $r=r_\star$. The red lines represent conformal geodesics with constant time leaving this initial hypersurface. The curves are computed by
setting $\lambda=3$ and $\phi=\frac{\pi}{4}$.  }
\label{Figure:PenroseDiagramConformalGeodesics}
\end{figure}

\subsection{Estimating the size of $D^+(\mathcal{R}_\bullet)$}
\label{Section:EstimatingDR}
Up to this point the size of the domain $\mathcal{R}_\bullet\subset
\mathcal{S}_\star$ (or more precisely, the value of the constant
$t_\bullet$ has remained unspecified). An inspection of the Penrose
diagram of the Schwarzschild-de Sitter spacetime shows that if the
value of $t_\bullet$ is too small, it could happen that the future
domain of dependence $D^+(\mathcal{R}_\bullet)$ is bounded and,
accordingly, will not reach the spacelike conformal boundary
$\mathscr{I}^+$ ---see e.g. Figure \ref{Figure:FutureDomainofDependence}. Given
our interest in constructing perturbations of the Schwarzschild-de
Sitter spacetime which contain as much as possible of the conformal boundary it
is then necessary to ensure that $t_\bullet$ is sufficiently large. In
this subsection given a fiduciary hypersurface $\mathcal{S}_\star$ in the
Cosmological region of the spacetime, we provide an estimate of how
large should $t_\bullet$ be for $D^+(\mathcal{R}_\bullet)$ to be
unbounded. In order to obtain this estimate we consider the future-oriented
inward-pointing null geodesics emanating from the end-points of
$\mathcal{R}_\bullet$ and look at where these curves intersect the
conformal boundary. 

\medskip
In order to carry out the analysis in this subsection it is convenient
to consider the coordinate $z\equiv 1/r$. In terms of this new coordinate, the line element
\eqref{BackgroundPhysicalMetric} takes the form
\[
  \mathring{\tilde{\bmg}} =\frac{1}{z^2} \bigg{(} -F(z)\mathbf{d} t  \otimes \mathbf{d} t +
\frac{1}{F(z)} \mathbf{d} z \otimes  \mathbf{d} z +
\bm{\sigma}\bigg{)},
\]
where
\[
  F(z)\equiv  z^2 D(1/z).
\]
The above expression suggest defining an \emph{unphysical metric} $\bar{\bmg}$ via
\[
  \bar{\bmg}= \Xi^2 \mathring{\tilde{\bmg}}, \qquad \Xi \equiv  z.
\]
More precisely, one has
\begin{equation}
  \bar{\bmg}=- F(z)\mathbf{d} t  \otimes \mathbf{d} t +
  \frac{1}{F(z)} \mathbf{d} z \otimes  \mathbf{d} z + \bm{\sigma}.
  \label{AuxiliaryConformalMetric}
\end{equation}

\medskip
In order to study the null geodesics we consider the Lagrangian
\[
  \mathcal{L}= -F(z) \dot{t}^2 + \frac{1}{F(z)} \dot{z}^2,
\]
where $\cdot \equiv \frac{d}{ds}$. In the case of null conformal geodesics $\mathcal{L}=0$ so that
\[
  \dot{t}= \pm \frac{1}{F(z)} \dot{z}.
\]
This, in turn, means that
\[
  \frac{dt}{dz} \dot{z}=\pm \frac{1}{F(z)} \dot{z}.
\]
By integrating both sides it follows that 
\[
  \int_{t_\bullet}^{t_+} dt= \pm \int_{z_\star}^{0} \frac{1}{F(z)}
  dz,
\]
where  $t_+$ denotes the value of the (spacelike) coordinate $t$ at
which the null geodesic reaches $\mathscr{I}^+$. Accordingly for the inward-pointing light rays emanating from the points on $\mathcal{S}_\star$
defined by the condition $t=t_\bullet$ one has that 
\begin{equation}
 t_+ = t_\bullet -  \int_{0}^{z_\star} \frac{1}{F(z)} dz.
\label{TerminalTime}
\end{equation}
An analogous condition holds for the inward-pointing light rays
emanating from the points with $t=-t_\bullet$. Since in the Cosmological region $F(z)>0$ it follows that
\[
 \int_{0}^{z_\star} \frac{1}{F(z)} dz > 0. 
 \]
The key observation in the analysis in this subsection is the following:
$D^+(\mathcal{R}_\bullet)$ is unbounded (so that it intersects the
conformal boundary) if $t_+$ as given by equation \eqref{TerminalTime}
satisfies $t_+>0$. As $t_\bullet>0$, having $t_+<0$ would mean that
the light rays emanating from the points with $t=t_\bullet$ and
$t=-t_\bullet$ intersect before reaching $\mathscr{I}^+$. Now, the
condition $t_+>0$ implies, in turn, that
\[
t_\bullet > \int_{0}^{z_\star} \frac{1}{F(z)} dz .
\]

As the integral in the right-hand side of the above inequality is not
easy to compute we provide, instead, a lower bound. One has then that
\[
  t_\bullet > \frac{z_\star}{F_\circledast},
\]
where $F_\circledast$ denotes the maximum of
\[
F(z)= z^2 -Mz^3 +1
\]
over the interval $[0,z_\star]$. Thus, $F'(z)$ vanishes if $z=0$ or
$z=z_\odot \equiv 2/3M$. Also, notice that $F'(z)>0$ for
$z\approx 0$. It can be readily verified that $F''(0)>0$ while
$F''(2/3M)<0$ so that an inflexion point occurs in the interval
$(0,z_\odot)$ and there are no other inflexion points in $[0,z_\star]$. Now, looking at the definition of $M$, equation \eqref{Horizon3}, and the
expression for $r_c$ as given by equation \eqref{DefinitionM} one concludes that
$z_\odot >z_c\equiv 1/r_c$. As $z_\odot$ is independent of
$z_\star$, it is not possible to decide whether $z_\odot$ lies in
$[0,z_\star]$ or not. In any case, one has that
\[
F(z_\odot)=1+ \frac{4}{27 M^2} \geq F_\circledast,
  \]
so that
\begin{equation}
  t_\bullet > \frac{27 M^2 z_\star}{27M^2 + 4}.
  \label{Boundtbullet}
\end{equation}

\medskip
One can summarise the discussion in this subsection as follows:

\begin{lemma}
  \label{Lemma:Boundtbullet}
If condition \eqref{Boundtbullet} holds then $D^+(\mathcal{R}_\bullet)$ is unbounded.
  \end{lemma}

  \begin{remark}
{\em In the rest of this article it is assumed that condition
  \eqref{Boundtbullet} always holds. }
\end{remark}

\begin{figure}[t]
\centering
\includegraphics[width=1\textwidth]{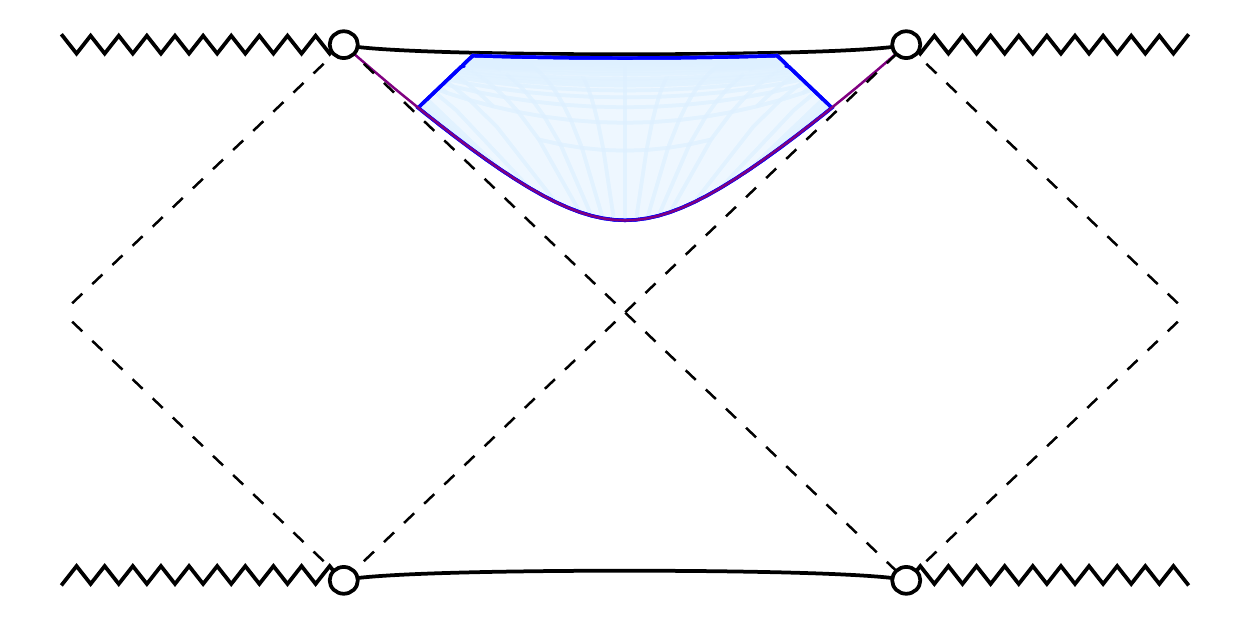}
\put(-310,205){$\mathcal{Q}$}
\put(-350,140){$r_b$}
\put(-300,160){$-t_\bullet$}
\put(-270,140){$r_c$}
\put(-220,205){$\mathscr{I}^+$}
\put(-220,165){$D^+(\mathcal{R}_\bullet)$}
\put(-160,140){$r_c$}
\put(-140,160){$t_\bullet$}
\put(-80,140){$r_b$}
\put(-120,205){$\mathcal{Q}'$}
\put(-270,70){$r_c$}
\put(-350,70){$r_b$}
\put(-220,1){$\mathscr{I}^-$}
\put(-160,70){$r_c$}
\put(-80,70){$r_b$}
\put(-310,1){$\mathcal{Q}$}
\put(-120,1){$\mathcal{Q}'$}
\caption{The plotted future domain of dependence of the solution
  $D^+(\mathcal{R}_\bullet)$ on the Penrose diagram of the 
  Cosmological region of the sub-extremal Schwarzschild-de Sitter
  spacetime. The value of $t_\bullet$ can be chosen as close as
  possible to the asymptotic points $\mathcal{Q}$ and $\mathcal{Q}'$
  so as to satisfy condition \eqref{Boundtbullet}.}
\label{Figure:FutureDomainofDependence}
\end{figure}
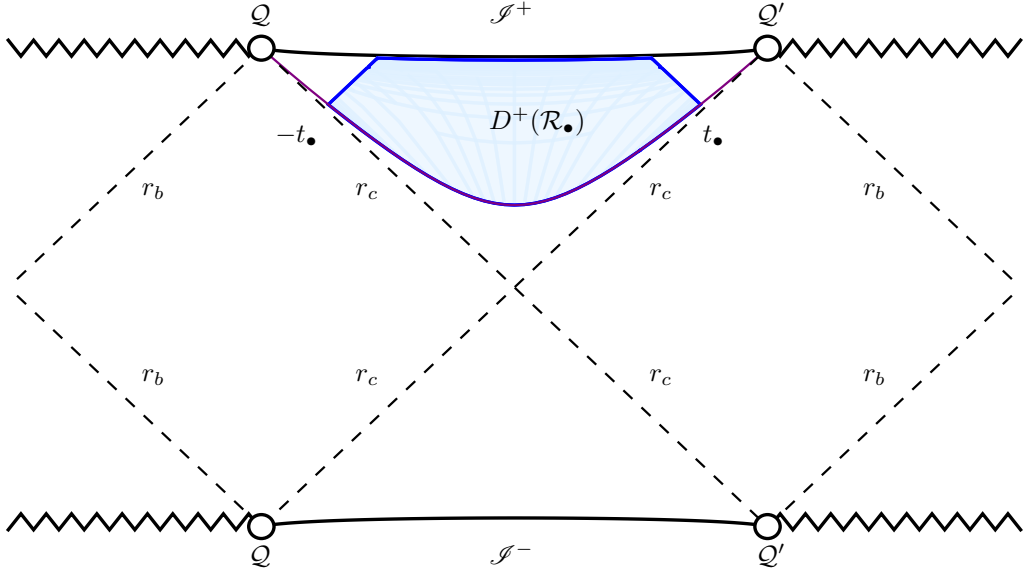

\subsection{Conformal Gaussian coordinates in the sub-extremal
  Schwarzschild-de Sitter spacetime}
\label{Subsection:CGCoordinates}
We now combine the results of the previous subsections to show that
the congruence of conformal geodesics defined by the initial
conditions \eqref{InitialData1} can be used to construct a
\emph{conformal Gaussian coordinate system} in a domain in the chronological
future of $\mathcal{R}_\bullet\subset \mathcal{S}_\star$,
$J^+(\mathcal{R}_\bullet\subset \mathcal{S}_\star)$, containing a
portion of the conformal boundary $\mathscr{I}^+$.

\medskip
In the following let $\widetilde{SdS}_I$ denote the Cosmological region of the
Schwarzschild-de Sitter spacetime ---that is
\[
\widetilde{SdS}_I =\{ p\in \tilde{\mathcal{M}} \; |\; r(p)>r_c \}.
\]
Moreover, denote by  $SdS_I$ the conformal representation of
$\widetilde{SdS}_I$ defined by the conformal factor $\Theta$ defined by
the non-singular congruence of conformal geodesics given by
Proposition \ref{Proposition:CGNoCaustics}. For $r >r_c$ let $z\equiv
1/r$ ---cfr the line element \eqref{AuxiliaryConformalMetric}. In terms of these coordinates, one has that
\begin{equation}
SdS_I = \{ p\in \mathbb{R}\times \mathbb{R}\times \mathbb{S}^2 \; |\;
0\leq z(p) \leq z_\star\}
\label{Definition:SdSI}
\end{equation}
where $z_\star\equiv 1/r_\star$ with $r_\star > r_c$. In particular, the
conformal boundary, $\mathscr{I}^+$, corresponds to the set of points
for which $z=0$. 

\medskip
The analysis of the previous subsections shows that the conformal
geodesics defined by the initial conditions \eqref{InitialData1} can
be thought of as curves on  $SdS_I$ of the form
\[
(\tau,t_\star) \mapsto \big( t(\tau,t_\star), z(\tau,t_\star),\theta_\star,\varphi_\star\big).
\]
Thus, in particular, the congruence of curves defines a map
\[
\psi:  [0,2] \times [-t_\bullet, t_\bullet] \rightarrow [0,z_\star] \times [-t_\bullet, t_\bullet].
\]
This map is analytic in the parameters $(\tau,t_\star)$. Moreover, the fact that the
congruence of conformal geodesics is non-intersecting implies that the
map is, in fact, invertible ---the analysis of the conformal geodesic
deviation equation implies that the Jacobian of the transformation is
non-zero for the given value of the parameters. In particular, it can be readily
verified that the function $\Theta \tilde{\omega}$ coincides with the
Jacobian of the transformation. Accordingly, the inverse map $\psi^{-1}$
\[
\psi^{-1}:  [0,z_\star]\times [-t_\bullet, t_\bullet]
\rightarrow [0,2] \times [-t_\bullet, t_\bullet], \qquad  (t,z) \mapsto \big(\tau(t,z),t_\star(t,z) \big)
\]
is well-defined. Thus, $\psi^{-1}$ gives the transformation from the
\emph{standard Schwarzschild
coordinates} $(t,z,\theta,\varphi)$ into the \emph{conformal Gaussian
coordinates} $(\tau, t_\star, \theta,\varphi)$. In the following let
\[
\mathcal{M}_\bullet \equiv [0,2] \times [-t_\bullet, t_\bullet].
\]
As the conformal geodesics of our congruence are timelike, we have
that 
\[
\mathcal{M}_\bullet \subset J^+(\mathcal{R}_\bullet).
\]
All throughout we assume, as discussed in Subsections
\ref{Subsection:InitialDataCongruence} and \ref{Section:EstimatingDR}, that
$t_\bullet$ is sufficiently large to ensure that
$D^+(\mathcal{R}_\bullet)$ contains a portion of $\mathscr{I}^+$
---cfr Lemma \ref{Lemma:Boundtbullet}.

\begin{proposition}
  \label{Proposition:CGCoordinates}
The congruence of conformal geodesics on $SdS_I$ defined by the
initial conditions on $\mathcal{S}_\star$ given by
\eqref{InitialData1} induce a conformal Gaussian coordinate system
over $D^+(\mathcal{R}_\bullet)$ which is related to the standard coordinates
$(t,r)$ via a map which is analytic. 
\end{proposition}

\section{The Schwarzschild-de Sitter spacetime in the conformal
  Gaussian system}
\label{Section:SdSGaussian}
In the previous section, we have established the existence of conformal
Gaussian coordinates in the domain $\mathcal{M}_\bullet \subset SdS_I$
of the Schwarzschild-de Sitter spacetime. In this section, we proceed
to analyse the properties of this exact solution in these
coordinates. This analysis is focused on the structural properties
relevant for the analysis of stability in the latter parts of this
article.

\begin{remark}
  {\em The metric coefficients implied by the line element \eqref{AuxiliaryConformalMetric} are
  analytic functions of the coordinates in the region
  $\mathcal{M}_\bullet$ ---barring the usual
  degeneracy of spherical coordinates.}
  \end{remark}

  \subsection{Weyl propagated frames}
  The ultimate aim of this section is to cast the Schwarzschild-de
  Sitter spacetime in the region $\mathcal{M}_\bullet$ as a solution
  to the extended conformal Einstein field equations introduced in
  Section \ref{Section:FrameXCFE}. A key step in this construction is the use of a Weyl
  propagated frame. In this section, we discuss a class of these frames
  in $\mathcal{M}_\bullet$.

  \medskip
  Since the congruence of conformal geodesics implied by the initial data
  \eqref{InitialData1} satisfies $\tilde{\bmbeta}=0$, the Weyl propagation equation
  \eqref{WeylPropagationGTildaAdapted}  reduces to the usual parallel
  propagation equation ---that is,
  \begin{equation}
\tilde{\nabla}_{\tilde{\bmx}'} (\Theta\tilde{\bme}_\bma) =
\tilde{\nabla}_{\tilde{\bmx}'} \bme_\bma=0.
\label{SimplifiedWeyl Propagation}
\end{equation}
The subsequent computations can be simplified by noticing that the
line element \eqref{BackgroundPhysicalMetric} is in warped-product
form. Given the spherical symmetry of the  Schwarzschild-de Sitter
spacetime, most of the discussion of a frame adapted to the symmetry of the
spacetime can be carried out by considering the 2-dimensional
Lorentzian metric
\begin{eqnarray*}
  && \bmell = \ell_{AB} \mathbf{d}x^A\otimes \mathbf{d}x^B\\
  && \phantom{\bmell} = 
      -D(r)\mathbf{d} t  \otimes \mathbf{d} t +
\frac{1}{D(r)} \mathbf{d} r \otimes  \mathbf{d} r.
  \end{eqnarray*}

\medskip
In the spirit of a conformal Gaussian system, we begin by setting the
\emph{time leg} of the frame as $\bme_\bmzero=\dot{\bmx}$. Then since
\[
  \dot{\bmx}=\Theta^{-1} \tilde{\bmx}',
\]
it follows that
\[
  \bme_\bmzero=\Theta^{-1} \tilde{\bmx}'.
\]
Now, recall that
\[
\tilde{\bmx}' = \tilde{t}'\bmpartial_t + \tilde{r}'\bmpartial_r,
\qquad \tilde{t}=t(\tilde{\tau}), \quad \tilde{r}=r(\tilde{\tau}),
\]
and let
\[
\bmomega\equiv \epsilon_\bmell(\tilde{\bmx}', \cdot ).
\]
It follows then that $\langle \bmomega, \tilde{\bmx}'\rangle =0$ so
that it is natural to consider a \emph{radial leg} of the frame,
$\bme_\bmone$, which is proportional to
$\bmomega^\sharp$. By using the condition
$\bmell(\bme_\bmone,\bme_\bmone)=1$ one readily finds that
\[
\bme_\bmone =\Theta \bmomega^\sharp.
\]

It can be readily verified by a direct computation that the vector
$\bme_\bmone$ as defined above satisfies the propagation equation
\eqref{SimplifiedWeyl Propagation}. 

Finally, the vectors $\bme_\bmtwo$ and $\bme_\bmthree$ are chosen in
such a way that they span the tangent space of the 2-spheres associated
to the orbits of the spherical symmetry. Accordingly, by setting
\[
\bme_\bmtwo = e_\bmtwo{}^{\mathcal{A}}\bmpartial_{\mathcal{A}}, \qquad
\bme_\bmthree = e_\bmthree{}^{\mathcal{A}}\bmpartial_{\mathcal{A}},
\qquad \mathcal{A}=2,\,3,
\]
it follows readily from the warped-product structure of the metric
that
\[
 \tilde{x}'^A (\partial_{A} e_{\bmtwo}{}^{\mathcal{A}})= \tilde{x}'^A (\partial_{A} e_{\bmthree}{}^{\mathcal{A}})=0.
\]
In other words, one has that the frame coefficients
$e_{\bmtwo}{}^{\mathcal{A}}$ and $e_{\bmthree}{}^{\mathcal{A}}$ are
constant along the conformal geodesics. Thus, in order to complete the
Weyl propagated frame $\{ \bme_\bma\}$ we choose \emph{two arbitrary orthonormal
vectors} $\tilde{\bme}_{\bmtwo\star}$ and $\tilde{\bme}_{\bmthree\star}$ spanning the
tangent space of $\mathbb{S}^2$ and define vectors $\{ \bme_\bmtwo,
\bme_\bmthree \}$ on $\mathcal{M}_\bullet$ by extending (constantly)
the value of the associated 
coefficients $\big(e_\bmtwo{}^{\mathcal{A}}\big)_\star$ and
$\big(e_\bmthree{}^{\mathcal{A}}\big)_\star$ along the conformal
geodesic. 

\medskip
The analysis of this subsection can be summarised in the following:

\begin{proposition}
  \label{Proposition:SdSWeylPropagatedFrame}
Let $\tilde{\bmx}'$ denote the vector tangent to the conformal
geodesics defined by the initial data \eqref{InitialData1} and let $\{
\bme_{\bmtwo\star},\, \bme_{\bmthree\star} \}$ be an arbitrary
orthonormal pair of vectors spanning the tangent bundle of
$\mathbb{S}^2$.
Then the frame $\{\bme_\bmzero,\,
\bme_\bmone,\,\bme_\bmtwo,\,\bme_\bmthree\}$ obtained by the procedure
 described in the previous paragraphs is a $\bmg$-orthonormal Weyl
propagated frame. The frame depends analytically on the unphysical
proper time $\tau$ and the initial position $t_\star$ of the curve. 
\end{proposition}

\begin{remark}
{\em In the previous proposition we ignore the usual complications due
to the non-existence of a globally defined basis of
$T\mathbb{S}^2$. The key observation is that any local choice works well. }
  \end{remark}

  \subsection{The Weyl connection}
  \label{Subsection:WeylConnection}
 The connection coefficients associated to a conformal Gaussian gauge
 are made up of two pieces: the 1-form defining the Weyl connection
 and the Levi-Civita connection of the metric $\bar{\bmg}$. We analyse
 these two pieces in turn. 

 \subsubsection{The 1-form associated to the Weyl connection}
We start by recalling that in Section \ref{Section:CGSdS} a congruence of conformal
geodesics with data prescribed on the hypersurface $\mathcal{S}_\star$
was considered. This congruence was analysed using the
$\tilde{\bmg}$-adapted conformal geodesic equations. The initial data
for this congruence was chosen so that the curves with tangent given
by $\tilde{\bmx}'$ satisfy the standard (affine) geodesic
equation. Consequently, the (spatial) 1-form
$\tilde{\bmbeta}$ vanishes. Thus, the 1-form $\bmbeta$ is given by
\[
  \bmbeta =- \dot{\Theta} \tilde{\bmx}^{\prime\flat},
  \]
  ---cfr. equation \eqref{CG:OneFormSplit}. Now, recalling that
  $\tilde{\bmx}' = r' \bmpartial_r$ and observing equation
  \eqref{rtildedash} one concludes that
  \[
\tilde{\bmx}^{\prime\flat} = \frac{1}{|\sqrt{D(r)}|} \mathbf{d}r.
\]
Rewritten in terms of $z$, the latter gives
\[
\tilde{\bmx}^{\prime\flat}  = -\frac{1}{z \sqrt{|F(z)|} } \mathbf{d}z.
\]
As $F(0)=1$, and $\dot{\Theta}|_{\mathscr{I}^+}=-1$ (cfr. equation \eqref{CanonicalThetaSdS}), it then follows
that 
\[
\bmbeta \approx -\frac{1}{z}\mathbf{d}z \qquad \mbox{for}\quad
z\approx 0.
\]
That is, $\bmbeta$ is singular at the conformal boundary. However, in the
subsequent analysis the key object is not $\bmbeta$ but
$\bar{\bmbeta}$, the 1-form associated to the conformal geodesics
equations written in terms of the connection $\bar{\bmnabla}$. Now, from
the conformal transformation rule $\bar{\bmbeta} = \bmbeta + \Xi^{-1}
\mathbf{d}\Xi $ and recalling that $\Xi=z$ it follows that  
\[
\bar{\bmbeta}= \frac{\dot{\Theta}}{z \sqrt{|F(z)|} }\mathbf{d}z + \frac{1}{z} \mathbf{d}z.
\]
Thus, from the preceding discussion it follows that $\bar{\bmbeta}$
is smooth at $\mathscr{I}^+$ and, moreover,
$\bar{\bmbeta}|_{\mathscr{I}^+}=0$. Notice, however, that
$\bar{\bmbeta}\neq 0$ away from the conformal boundary.

\subsubsection{Computation of the connection coefficients}
The 1-form $\bmbeta$ defines in a natural way a Weyl connection $\hat{\bmnabla}$ via
the relation
\[
\hat{\bmnabla}-\tilde{\bmnabla} = \mathbf{S}(\bmbeta)
\]
where $\mathbf{S}$ corresponds to the tensor $S_{ab}{}^{cd}$ as
defined in \eqref{WeylToUnphysical}. As the coordinates and
connection coefficients associated to the physical connection
$\tilde{\bmnabla}$ are not well adapted to a discussion near the
conformal boundary we resort to the unphysical Levi-Civita connection
$\bar{\bmnabla}$ to compute $\hat{\bmnabla}$. From the discussion in the
previous subsections, we have that
\[
\bar{\bmnabla}-\tilde{\bmnabla} =\mathbf{S}(z^{-1}\mathbf{d}z).
\]
It thus follows that
\[
\hat{\bmnabla}-\bar{\bmnabla} =\mathbf{S}(\bar{\bmbeta}).
\]

\medskip
Now let $\{ \bme_\bma\}$ denote the Weyl propagated frame as given by
Proposition \ref{Proposition:SdSWeylPropagatedFrame}. The connection
coefficients $\hat{\Gamma}_\bma{}^\bmb{}_\bmc$ are defined through the
relation
\[
\hat{\nabla}_\bma \bme_\bmc = \hat{\Gamma}_\bma{}^\bmb{}_\bmc \bme_\bmb.
\]
Now, writing $\bme_\bma=e_\bma{}^\mu\bmpartial_\mu$ one has that
\[
\hat{\nabla}_\bma \bme_\bmc = \big(\hat{\nabla}_\mu e_\bmc{}^\nu\big)
e_\bma{}^\mu \bmpartial_\nu,
\]
where
\begin{eqnarray}
&& \hat{\nabla}_\mu e_\bmc{}^\nu = \bar{\nabla}_\mu e_\bmc{}^\nu +
   S_{\mu\lambda}{}^{\nu\rho}\bar{\beta}_\rho e_\bmc{}^\lambda,\nonumber\\
   && \phantom{\hat{\nabla}_\mu e_\bmc{}^\nu}= \partial_\mu e_\bmc{}^\nu+
      \bar{\Gamma}_\mu{}^\nu{}_\lambda e_\bmc{}^\lambda +
      S_{\mu\lambda}{}^{\mu\rho} \bar{\beta}_\rho e_\bmc{}^\lambda. \label{ComponentsWeylCovDer}
\end{eqnarray}

\medskip
A direct computation shows that the only non-vanishing
Christoffel symbols of
the metric \eqref{AuxiliaryConformalMetric},
$\bar{\Gamma}_\mu{}^\nu{}_\lambda$ are given by
\begin{eqnarray*}
&&\bar{\Gamma}_t{}^t{}_z = -\bar{\Gamma}_z{}^z{}_z = \frac{z
  (\frac{3}{2}Mz-1)}{1 + z^2(Mz-1)}, \\
&&  \bar{\Gamma}_t{}^z{}_t = z(\tfrac{3}{2}Mz-1)\big( 1
  +z^2(Mz-1) \big),\\
&& \bar{\Gamma}_\varphi{}^\theta{}_\varphi =-\cos\theta \sin\theta, \qquad
   \bar{\Gamma}_\theta{}^\varphi{}_\varphi=\cot \theta.  
  \end{eqnarray*}
Observe that the coefficients $\bar{\Gamma}_t{}^t{}_z$,
$\bar{\Gamma}_z{}^z{}_z$ and $
\bar{\Gamma}_t{}^z{}_t$ are analytic at $z=0$.

\begin{remark}
{\em  The connection coefficients $\bar{\Gamma}_\varphi{}^\theta{}_\varphi$, $\bar{\Gamma}_\theta{}^\varphi{}_\varphi$
correspond to the connection of the round metric over
$\mathbb{S}^2$. In the rest of this section, we ignore this coordinate
singularity due to the use of spherical coordinates.}
\end{remark}

\medskip
It follows from the discussion in the previous paragraphs and
Proposition \ref{Proposition:SdSWeylPropagatedFrame} that each of the
terms in the righthand side of \eqref{ComponentsWeylCovDer} is a
regular function of the coordinate $z$ and, in particular, analytic at
$z=0$. Contraction with the coefficients of the frame does not change
this. Accordingly, it follows that the Weyl connection coefficients
$\hat{\Gamma}_\bma{}^\bmb{}_\bmc$ are smooth functions of the
coordinates used in the conformal Gaussian gauge on the future of the
fiduciary initial hypersurface $\mathcal{S}_\star$ up to and beyond
the conformal boundary.

\subsection{The components  of the curvature}
In this section we discuss the behaviour of the various components of
the curvature of the Schwarzschild-de Sitter spacetime in the domain
$\mathcal{M}_\bullet$. We are particularly interested in the behaviour
of the curvature at the conformal boundary.

\medskip
The subsequent discussion is best done in terms of the conformal
metric $\bar{\bmg}$ as given by
\eqref{AuxiliaryConformalMetric}. Consider also the vector
$\bar{\bme}_\bmzero$ given by
\[
 \bar{\bme}_\bmzero= \sqrt{|F(z)|}\bmpartial_z, \qquad  F(z)= z^2 -Mz^3- 1.
\]
This vector is orthogonal to the conformal boundary $\mathscr{I}^+$
which, in these coordinates is given by the condition $z=0$.

\subsubsection{The rescaled Weyl tensor}
Given a timelike vector, the components of the rescaled Weyl tensor
$d_{abcd}$ can be conveniently encoded in the electric and magnetic parts relative
to the given vector. For the vector $ \bar{\bme}_\bmzero$ these are
given by
\[
 d_{ac}=d_{abcd}\bar{e}_\bmzero{}^b \bar{e}_\bmzero {}^{d}, \qquad d^{*}{}_{ac}=d{}^*{}_{abcd}\bar{e}_{\bmzero}{}^{b}\bar{e}_\bmzero{}^{d},
\]
where $d{}^*{}_{abcd}$ denotes the Hodge dual of $d_{abcd}$. A
computation using the package {\tt xAct} for {\tt Mathematica} readily
gives that the only non-zero components of the electric part are given by
\begin{eqnarray*}
&&d_{tt}= -M\big{(} z^2(1-Mz)-1 \big{)}, \\
&& d_{\theta\theta}=-\frac{M}{2},\\
&& d_{\varphi\varphi}=-\frac{M}{2}\sin^2\theta,  
\end{eqnarray*}
while the magnetic part vanishes identically. Observe, in particular,
that the above expressions are regular at $z=0$ ---again, disregarding
the coordinate singularity due to the use of spherical
coordinates. The smoothness of the components of the Weyl tensor is
retained when re-expressed in terms of the Weyl propagated frame
$\{\bme_\bma\}$ as given in Proposition \ref{Proposition:SdSWeylPropagatedFrame}.

\subsubsection{The Schouten tensor}
A similar computer algebra calculation shows that the non-zero
components of the Schouten tensor of the metric $\bar{\bmg}$ are given by
\begin{eqnarray*}
&&\bar{L}_{tt}=\frac{1}{2} (2Mz-1)(1 +  z^2(Mz-1)),\\
&&\bar{L}_{zz}=-\frac{1}{2}\frac{(2Mz-1)}{1 +z^2(Mz-1)},\\
&&\bar{L}_{\theta\theta}=-\frac{1}{2}(Mz-1),\\
  && \bar{L}_{\varphi\varphi}=-\frac{1}{2} \sin^2\theta(Mz-1).
\end{eqnarray*}
Again, disregarding the coordinate singularity on the angular
components, the above expressions are analytic on
$\mathcal{M}_\bullet$ ---in particular at $z=0$. To obtain the
components of the Schouten tensor associated to the Weyl connection
$\hat{\nabla}$ we make use of the transformation rule
\[
\bar{L}_{ab}-\hat{L}_{ab} = \bar{\nabla}_a \bar{\beta}_b - \frac{1}{2}S_{ab}{}^{cd}\bar{\beta}_c\bar{\beta}_d.
\]
The smoothness of $\bar{\beta}_a$ has already been established in
Subsection \ref{Subsection:WeylConnection}. It follows then that the
components of $\hat{L}_{ab}$ with respect to the Weyl propagated frame
$\{\bme_\bma\}$ are regular on $\mathcal{M}_\bullet$.  

\subsection{Summary}
The analysis of the preceding subsections is summarised in the
following:

\begin{proposition}
  \label{Proposition:SummarySdSConformalGaussian}
Given $t_\bullet>0$ and the Weyl propagated frame $\{\bme_\bma\}$ as given by Proposition
\ref{Proposition:SdSWeylPropagatedFrame}, the connection coefficients
of the Weyl connection associated to the congruence of conformal
geodesics, the components of the rescaled Weyl tensor and the
components of the Schouten tensor of the Weyl connection are smooth on
$\mathcal{M}_\bullet$ and, in particular, at the conformal boundary.
\end{proposition}

\begin{remark}
{\em In other words, the sub-extremal Schwarzschild-de Sitter spacetime
expressed in terms of a conformal Gaussian gauge system gives rise to
a solution to the extended conformal Einstein field equations on the
region $\mathcal{M}_\bullet \subset D^+(\mathcal{R}_\bullet)$ where
$\mathcal{R}_\bullet\subset \mathcal{S}_\star$. }
  \end{remark}

\subsection{Construction of a background solution with compact spatial
  sections}
The region $\mathcal{R}_\bullet\subset \mathcal{S}_\star$ has the
topology of $I\times \mathbb{S}^2$ where $I\subset \mathbb{R}$ is an
open interval.  Accordingly, the spacetime arising from
$\mathcal{R}_\bullet$ will have spatial sections with the same
topology. As part of the perturbative argument given in Section \ref{Section:ExistenceUniquenesStability}
based on the general theory of symmetric hyperbolic systems as given
in \cite{Kat75b} it is convenient to consider solutions with compact
spatial sections. We briefly discuss how the (conformal)
Schwarzschild-de Sitter spacetime in the conformal Gaussian system
over $\mathcal{M}_\bullet$ can
be recast as a solution to the extended conformal Einstein field
equations with compact spatial sections.

\medskip
\noindent
The key observation on this construction is that the Killing vector
$\bmxi=\bmpartial_t$ in the Cosmological region of the spacetime is
spacelike. Thus, given a fixed $z_\circ<z_c$, we have that the
hypersurface $\mathcal{S}_{z_\circ}$ defined by the condition
$z=z_\circ$ has a translational invariance
---that is, the intrinsic metric $\bmh$ and the extrinsic curvature
$\bmK$ are invariant under the replacement $t\mapsto t + \varkappa$
for $\varkappa\in \mathbb{R}$. Moreover, the congruence of conformal
geodesics given by Proposition
\ref{Proposition:SummarySdSConformalGaussian} are such that the value
of the coordinate $t$ is constant along a given curve.

Consider now, the timelike hypersurfaces $\mathcal{T}_{-2t_\bullet}$
and $\mathcal{T}_{2t_\bullet}$ in $D^+(\mathcal{S}_\star)$ generated,
respectively, by the future-directed geodesics emanating from
$\mathcal{S}_\star$ at the points with $t=-2t_\bullet$ and
$t=2t_\bullet$. From the discussion in the previous paragraph, one can
identify $\mathcal{T}_{-2t_\bullet}$ and $\mathcal{T}_{2t_\bullet}$ to
obtain a smooth spacetime manifold $\bar{\mathcal{M}}_\bullet$ with
compact spatial sections ---see Figure
\ref{Figure:FutureDomainofDependenceWithExtension}. A natural foliation
of $\bar{\mathcal{M}}_\bullet$ is given by the hypersurfaces
$\bar{\mathcal{S}}_z$ of constant $z$ with $0\leq z\leq z_\star$
having the topology of a 3-handle ---that is, $\mathcal{H}_z\approx
\mathbb{S}^1\times \mathbb{S}^2$.

The metric $\bar{\bmg}$ on $SdS_I$, cfr \eqref{Definition:SdSI}, induces a metric on
$\bar{\mathcal{M}}_\bullet$ which, by an abuse of notation, we denote
again by $\bar{\bmg}$. As the initial conditions defining the
congruence of conformal geodesics of Proposition
\ref{Proposition:CGNoCaustics} have translational invariance, it
follows that the resulting curves also have this property. Accordingly,
the congruence of conformal geodesics on $SdS_I$ given by Proposition
\ref{Proposition:CGNoCaustics} induces a non-intersecting congruence
of conformal geodesics on $\bar{\mathcal{M}}_\bullet$ ---recall that
each of the curves in the congruence has constant coordinate $t$.

\medskip
In summary, it follows from the discussion in the preceding
paragraphs that the solution to the extended conformal Einstein field
equations in a conformal Gaussian gauge as given by Proposition
\ref{Proposition:SummarySdSConformalGaussian} implies a similar
solution over the manifold $\bar{\mathcal{M}}_\bullet$. In the
following, we will denote this solution by $\mathring{\mathbf{u}}$. The
initial data induced by $\mathring{\mathbf{u}}$ on
$\bar{\mathcal{S}}_\star$ will be denoted by
$\mathring{\mathbf{u}}_\star$.

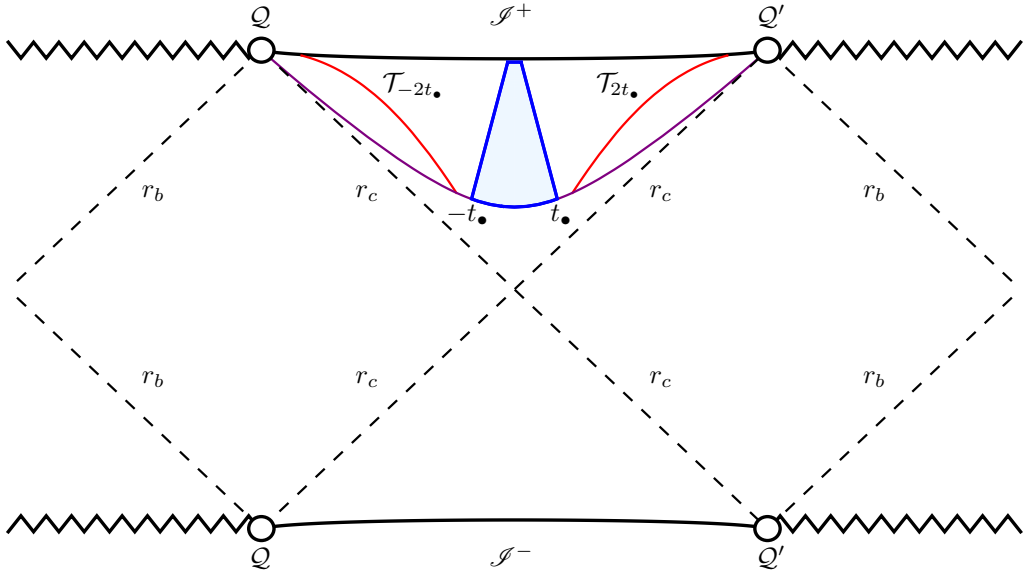
\begin{figure}[t]
\centering
\includegraphics[width=1\textwidth]{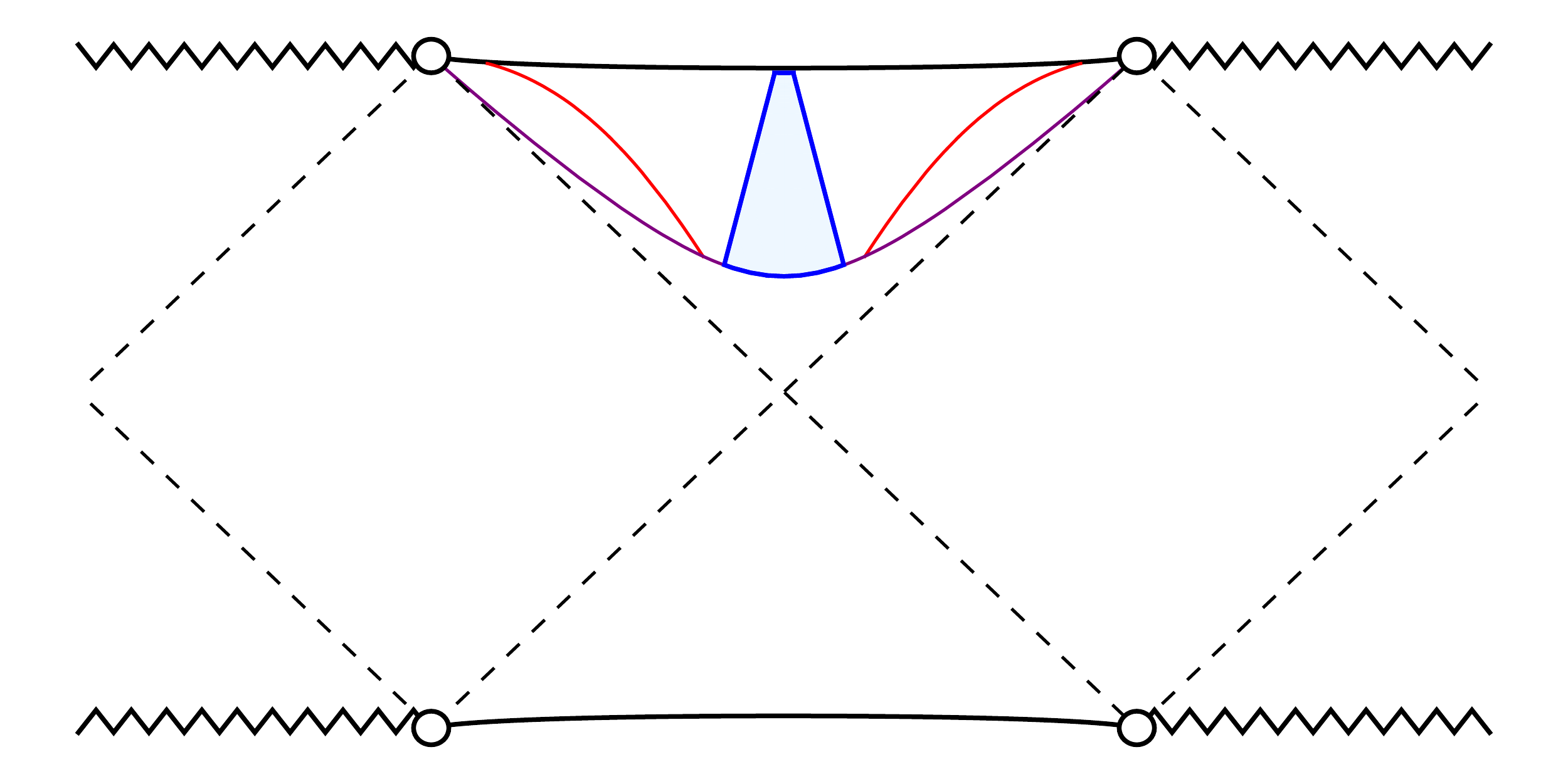}
\put(-310,205){$\mathcal{Q}$}
\put(-350,140){$r_b$}
\put(-236,131){$-t_\bullet$}
\put(-270,140){$r_c$}
\put(-220,205){$\mathscr{I}^+$}
\put(-160,140){$r_c$}
\put(-197,131){$t_\bullet$}
\put(-80,140){$r_b$}
\put(-120,205){$\mathcal{Q}'$}
\put(-270,70){$r_c$}
\put(-350,70){$r_b$}
\put(-220,1){$\mathscr{I}^-$}
\put(-160,70){$r_c$}
\put(-80,70){$r_b$}
\put(-310,1){$\mathcal{Q}$}
\put(-120,1){$\mathcal{Q}'$}
\put(-260,180){$\mathcal{T}_{-2t_\bullet}$}
\put(-180,180){$\mathcal{T}_{2t_\bullet}$}
\caption{The red curves identify the timelike hypersurfaces
  $\mathcal{T}_{-2t_\bullet}$ and $\mathcal{T}_{2t_\bullet}$. The resulting spacetime manifold $\bar{M}_\bullet$ has compact spatial
  sections, $\bar{\mathcal{S}}_z$, with the topology of $\mathbb{S}^1\times\mathbb{S}^2$.}
\label{Figure:FutureDomainofDependenceWithExtension}
\end{figure}

\section{The construction of non-linear perturbations }
\label{Section:ExistenceUniquenesStability}

In this section, we bring together the analysis carried out in the
previous sections to construct non-linear perturbations of the
Schwarzschild-de Sitter spacetime on a suitable portion of the
Cosmological region. 

\subsection{Initial data for the evolution equations}
Given a solution $(\mathcal{S}_\star, \tilde{\bmh}, \tilde{\bmK})$ to
the Einstein constraint equations, there exists an algebraic procedure
to compute initial data for the conformal evolution equations ---see
\cite{CFEBook}, Lemma 11.1. In the following, it will be assumed that we
have at our disposal a family of initial data sets for the vacuum
Einstein field equations corresponding to perturbations of initial
data for the Schwarzschild-de Sitter spacetime on hypersurfaces of constant
coordinate $r$ in the Cosmological region. Initial data for the
conformal evolution equations can then be constructed out of these
basic initial data sets. \emph{Assumptions of this type
  are standard in the analysis of non-linear stability. }

\begin{remark}
{\em An interesting open problem is that of the construction of
  perturbative initial data sets for the evolution problem considered
  in this article using the Friedrich-Butscher method ---see
  e.g. \cite{But06,But07,ValWil20}. In this setting the free data is
  associated to a pair of rank 2 transverse and trace-free tensors
 prescribing suitable components of the curvature (i.e. the Weyl
 tensor) on the initial hypersurface. The main technical difficulty
 in this approach is the analysis of the Kernel of the linearisation
 of the so-called extended Einstein constraint equations. }
  \end{remark}

\medskip
Given a compact hypersurface $\bar{\mathcal{S}}_z\approx
\mathbb{S}^1\times\mathbb{S}^2$ and a function $\mathbf{u}:
\bar{\mathcal{S}}_z \rightarrow \mathbb{R}^N$ let $||
\mathbf{u}||_{\bar{\mathcal{S}}_z,m}$ for $m\geq 0$ denote the
standard $L^2$-Sobolev norm of order $m$ of $\mathbf{u}$. Moreover,
denote by $H^m(\bar{\mathcal{S}}_z,\mathbb{R}^N)$ the associated
Sobolev space ---i.e. the completion of the functions $\mathbf{w}\in
C^\infty(\bar{\mathcal{S}}_z,\mathbb{R}^N)$ under the norm $|| \phantom{
\mathbf{u}}||_{\bar{\mathcal{S}}_z,m}$. 

\medskip
In the following, consider some initial data set for the conformal
evolution equations $\mathbf{u}_\star$ on $\mathcal{R}_\bullet\approx
[-t_\bullet,t_\bullet]\times \mathbb{S}^2$ which
is a small perturbation of exact data $\mathring{\mathbf{u}}_\star$
for the Schwarzschild-de Sitter spacetime in the sense that
\[
\mathbf{u}_\star
=\mathring{\mathbf{u}}_\star+\breve{\mathbf{u}}_\star, \qquad ||\breve{\mathbf{u}}_\star||_{\mathcal{R}_\bullet,m}<\varepsilon
\]
for $m\geq 4$ and some suitably small $\varepsilon>0$. Making use of a
smooth cut-off function over $\bar{\mathcal{S}}_{z_\star}\approx
\mathbb{S}^1\times\mathbb{S}^2$ the perturbation data
$\breve{\mathbf{u}}_\star$ over $\mathcal{R}_\bullet$ can be matched
to vanishing data $\mathbf{0}$ on
$[-2t_\bullet,-\frac{3}{2}t_\bullet]\times\mathbb{S}^2\cup
[\frac{3}{2}t_\bullet, 2t_\bullet]\times\mathbb{S}^2$ with a smooth transition
region, say, 
$[-\frac{3}{2}t_\bullet,-t_\bullet]\times\mathbb{S}^2\cup[t_\bullet,\frac{3}{2}t_\bullet]\times\mathbb{S}^2$. In
this way one can obtain a vector-valued function
$\breve{\bar{\mathbf{u}}}_\star$ over
$\bar{\mathcal{S}}_\star\approx\mathbb{S}^1\times\mathbb{S}^2$ whose
size is controlled by the perturbation data $\breve{\mathbf{u}}_\star$
on $\mathcal{R}_\bullet$. In a slight abuse of notation, in order to
ease the reading, we write $\breve{\mathbf{u}}_\star$ rather than $\breve{\bar{\mathbf{u}}}_\star$.

\subsection{Structural properties of the evolution equations}
\label{Section:EvolutionEqns}
In this section, we briefly review the key structural properties of the
evolution system associated to the
extended conformal Einstein equations \eqref{ecfe5} written in terms
of a conformal Gaussian system. This evolution system
is central in the discussion of the stability of the background
spacetime. In addition, we also
discuss the subsidiary evolution system satisfied by the
zero-quantities associated to the field equations, \eqref{ecfe1}-\eqref{ecfe4},
and the supplementary zero-quantities \eqref{Supplementary1}-\eqref{Supplementary3}. The
subsidiary system is key in the analysis of the so-called \emph{propagation of the
constraints} which allows to establish the relation between a solution
to the extended conformal Einstein equations \eqref{ecfe5} and the Einstein
field equations \eqref{EFE}. One of the advantages of the hyperbolic reduction of the extended conformal
Einstein field equations by means of conformal Gaussian systems is that it provides
a priori knowledge of the location of the conformal boundary of the solutions
to the conformal field equations. 

\medskip
Conformal Gaussian gauge systems lead to a \emph{hyperbolic reduction}
of the extended conformal Einstein field equation \eqref{ecfe5}. The
particular form of the resulting evolution equations will not be
required in the analysis, only general structural properties. In order
to describe these denote by ${\bmupsilon}$ the independent
components of the coefficients of the frame $e_\bma{}^\mu$, the
connection coefficients $\hat{\Gamma}_\bma{}^\bmb{}_\bmc$ and the Weyl
connection Schouten tensor $\hat{L}_{\bma\bmb}$ and by $\bmphi$ the
independent components of the rescaled Weyl tensor
$d_{\bma\bmb\bmc\bmd}$, expressible in terms of its electric and
magnetic parts with respect to the timelike vector
$\bme_\bmzero$. Also, let $\bme$ and ${\bmGamma}$ denote,
respectively, the independent components of the frame and
connection. In terms of these objects one has the following:

\begin{lemma}
  \label{Lemma:EvolutionEqns}
The extended conformal Einstein field equations \eqref{ecfe5}
expressed in in terms of a conformal Gaussian gauge imply a symmetric
hyperbolic system for the components $({\bmupsilon},\bmphi)$ of
the form
\begin{subequations}
\begin{eqnarray}
&& \partial {\bmupsilon} = \mathbf{K} {\bmupsilon} +
   \mathbf{Q}({\bmGamma}){\bmupsilon} + \mathbf{L}(\bar{x}) \bmphi, \label{EvolutionEqn1}\\
  && \big( \mathbf{I} + \mathbf{A}^0(\bme)  \big)\partial_\tau \bmphi
     + \mathbf{A}^\alpha(\bme) \partial_\alpha \bmphi
     =\mathbf{B}({\bmGamma})\bmphi, \label{EvolutionEqn2}
     \end{eqnarray}
     \end{subequations}
where $\mathbf{I}$ is the unit matrix, $\mathbf{K}$ is a constant
matrix $\mathbf{Q}({\bmGamma})$ is a smooth matrix-valued
function, $\mathbf{L}(\bar{x})$ is a smooth matrix-valued
function of the coordinates, $\mathbf{A}^\mu(\bme)$ are Hermitian
matrices depending smoothly on the frame coefficients and
$\mathbf{B}({\bmGamma})$ is a smooth matrix-valued function of the
connection coefficients.
\end{lemma}

\begin{remark}
{\em In this article we will be concerned with situations in which the
matrix-valued function $ \mathbf{I} + \mathbf{A}^0(\bme)$ is positive
definite. This is the case, for example, in perturbations of a
background solution.}
\end{remark}

\begin{remark}
{\em Explicit expressions of the evolution equations and further
  discussion on their derivation can be found in
  \cite{MinVal21}  ---see also \cite{CFEBook}, Section 13.4 for a
  spinorial version of the equations.}
\end{remark}

For the evolution system \eqref{EvolutionEqn1}-\eqref{EvolutionEqn2}
one has the following \emph{propagation of the constraints} result \cite{MinVal21}:

\begin{lemma}
  \label{Lemma:PropagationConstraints}
Assume that the evolution equations
\eqref{EvolutionEqn1}-\eqref{EvolutionEqn2} hold. Then the independent
components of the zero-quantities
\[
{\Sigma}_\bma{}^\bmb{}_\bmc, \quad
{\Xi}^\bmc{}_{\bmd\bma\bmb}, \quad {\Delta}_{\bma\bmb\bmc},
\quad \Lambda_{\bma\bmb\bmc}, \quad \delta_\bma,\quad
\gamma_{\bma\bmb}, \quad \varsigma_{\bma\bmb},
\]
not determined by either the evolution equations or the gauge
conditions satisfy a symmetric hyperbolic system which is homogeneous
in the zero-quantities. As a result, if the zero-quantities vanish on
a fiduciary spacelike hypersurface $\mathcal{S}_\star$, then they also
vanish on the domain of dependence.
  \end{lemma}

  \begin{remark}
{\em It follows from Lemmas \ref{Lemma:EvolutionEqns}, \ref{Lemma:PropagationConstraints} and \ref{Lemma:XCFEtoEFE} that a solution to the conformal
  evolution equations \eqref{EvolutionEqn1}-\eqref{EvolutionEqn2}
  with data on $\mathcal{S}_\star$ satisfying the conformal
  constraints implies a solution to the Einstein field equations away
  from the conformal boundary. }
    \end{remark}

\subsection{Setting up the perturbative existence argument}
In the spirit of the schematic notation used in the previous section,
we set $\mathbf{u} \equiv (\bmv,\bmphi)$. Moreover, consistent with this
notation let $\mathring{\mathbf{u}}$ denote a solution to the evolution
equations \eqref{EvolutionEqn1} and \eqref{EvolutionEqn2} arising from some data
$\mathring{\mathbf{u}}_\star$ prescribed on a hypersurface at
$r=r_\star$. We refer to $\mathring{\mathbf{u}}$ as the \emph{background
  solution}. We will construct solutions to \eqref{EvolutionEqn1} and
\eqref{EvolutionEqn2}  which can be regarded as a perturbation of the
background solution in the sense that
\[
\mathbf{u}= \mathring{\mathbf{u}} + \breve{\mathbf{u}}.
\]
This means, in particular, that one can write
\begin{equation}
  \label{split}
  \bme=\mathring{\bme} + \breve{\bme}, \qquad \bmGamma=\mathring{\bmGamma} + \breve{\bmGamma}, \qquad \bmphi=\mathring{\bmphi} + \breve{\bmphi}.
\end{equation}
The components of $\breve{\bme}$, $\breve{\bmGamma}$ and
$\breve{\bmphi}$ are our unknowns. Making use of the decomposition
\eqref{split} and exploiting that  $\mathring{\bmu}$ is a solution to
the conformal evolution equations one obtains the equations 
\begin{subequations}
\begin{align}
&\label{he3}\partial_\tau \breve{\bmupsilon}= \mathbf{K} \breve{\bmupsilon} + \mathbf{Q}(\mathring{\bmGamma}+ \breve{\bmGamma}) \breve{\bmupsilon}+ \mathbf{Q}(\breve{\bmGamma})\mathring{\bmupsilon} + \mathbf{L}(\bar{x}) \breve{\bmphi}+ \mathbf{L}(\bar{x}) \mathring{\bmphi},  \\
&\label{he4} (\mathbf{I} + \mathbf{A}^0(\mathring{\bme} + \breve{\bme})) \partial_\tau \breve{\bmphi}+ \mathbf{A}^\alpha(\mathring{\bme} + \breve{\bme}) \partial_\alpha \breve{\bmphi}=\mathbf{B}(\mathring{\bmGamma}+ \breve{\bmGamma}) \breve{\bmphi}+\mathbf{B}(\mathring{\bmGamma}+ \breve{\bmGamma}) \mathring{\bmphi}.
\end{align}
\end{subequations}
Now, it is convenient to define
\[
\bar{\mathbf{A}}^0(\tau, \underline{x}, \breve{\mathbf{u}}) \equiv \begin{pmatrix*}
      \mathbf{I} &  0 \\
      0 & \mathbf{I} + \mathbf{A}^0(\mathring{\bme}+ \breve{\bme})
     \end{pmatrix*},
\qquad 
\bar{\mathbf{A}}^\alpha(\tau, \underline{x}, \breve{\mathbf{u}}) \equiv \begin{pmatrix*}
  0 &   0 \\
      0 &   \mathbf{A}^\alpha(\mathring{\bme}+ \breve{\bme})
     \end{pmatrix*},
   \]
  and
\[
\bar{\mathbf{B}}(\tau, \underline{x}, \breve{\mathbf{u}})\equiv \breve{\mathbf{u}}\bar{\mathbf{Q}}\breve{\mathbf{u}}+ \bar{\mathbf{L}}(\bar{x})\breve{\mathbf{u}}+ \bar{\mathbf{K}}\breve{\mathbf{u}},
\] 
where 
\[
\breve{\mathbf{u}}\bar{\mathbf{Q}}\breve{\mathbf{u}} \equiv \begin{pmatrix*}
    \breve{\bmupsilon}\mathbf{Q}\breve{\bmupsilon} &     0 \\
     0 &
     \mathbf{B}(\breve{\bmGamma})\breve{\bmphi} +  \mathbf{B}(\breve{\bmGamma})\mathring{\bmphi}
     \end{pmatrix*},
\qquad 
\bar{\mathbf{L}}(\bar{x})\breve{\mathbf{u}} \equiv \begin{pmatrix*}
  \mathring{\bmupsilon}\mathbf{Q}\breve{\bmupsilon}+ \mathbf{Q}(\breve{\bmGamma})\mathring{\bmupsilon}&
    \mathbf{L}(\bar{x})\breve{\bmphi}+ \mathbf{L}(\bar{x})\mathring{\bmphi}  \\
      0 &  0
     \end{pmatrix*},
\]
\[
\bar{\mathbf{K}} \breve{\mathbf{u}}\equiv \begin{pmatrix*}
 \mathbf{K} \breve{\bmupsilon} &  0  \\
      0 & \mathbf{B}(\mathring{\Gamma})\breve{\bmphi}+ \mathbf{B}(\mathring{\Gamma})\mathring{\bmphi}
     \end{pmatrix*},
\]
denote, respectively, expressions which are quadratic, linear and constant terms in the unknowns. 

\medskip
In terms of the above expressions it is possible to rewrite the system
\eqref{he3}-\eqref{he4} in the more concise form 
\begin{equation}
 \bar{\mathbf{A}}^0(\tau, \underline{x}, \breve{\mathbf{u}})\partial_\tau \breve{\mathbf{u}}+ \bar{\mathbf{A}}^\alpha(\tau, \underline{x}, \breve{\mathbf{u}})\partial_\alpha \breve{\mathbf{u}}=\bar{\mathbf{B}}(\tau, \underline{x}, \breve{\mathbf{u}}). \label{he5}
\end{equation}
These equations are in a form where the theory of first order
symmetric hyperbolic systems can be applied to obtain a existence and
stability result for small perturbations of the initial data
$\mathring{\mathbf{u}}_\star$. This requires, however, the introduction of the
appropriate norms measuring the size of the perturbed initial data
$\breve{\mathbf{u}}_\star$.

\begin{remark}
{\em In the following it will be assumed that the background solution
  $\mathring{\mathbf{u}}$ is given by the Schwarzschild-de Sitter
  background solution written in a conformal Gaussian gauge system as
  described in Proposition
  \ref{Proposition:SummarySdSConformalGaussian}. It follows that the
  entries of $\mathring{\mathbf{u}}$ are smooth functions on
  $\bar{\mathcal{M}}_\bullet\equiv [0,2]\times \bar{\mathcal{S}}_\star
  \approx [0,2]\times \mathbb{S}^1 \times\mathbb{S}^2$.} 
\end{remark}

\medskip
\begin{theorem} [\textbf{\em existence and uniqueness of the
    solutions to the conformal evolution equations}]
\label{Theorem:ExistenceConformalEvolution}
 Given $\mathbf{u}_\star=\mathring{\mathbf{u}}_\star
 +\breve{\mathbf{u}}_\star $ satisfying the conformal constraint
 equations on $\bar{\mathcal{S}}_\star$ 
 and  $m \geq 4$, one has that:
\begin{itemize}
\item[(i)] There exists $\varepsilon >0$ such that if 
\begin{equation}
||\breve{\mathbf{u}}_\star||_{\bar{\mathcal{S}}_\star,m} < \varepsilon, \label{SizeData}
\end{equation}
then there exists a unique solution $\breve{\mathbf{u}}\in C^{m-2}(
[0,2]\times \bar{\mathcal{S}}_\star, \mathbb{R}^N)$ to the Cauchy problem for the
conformal evolution equations \eqref{he5} with initial data
$\breve{\mathbf{u}}(0,\underline{x})=\breve{\mathbf{u}}_\star$ and with
$N$ denoting the dimension of the vector $\breve{\mathbf{u}}$.

\item[(ii)] Given a sequence of initial data $
  \breve{\mathbf{u}}{}^{(n)}_\star$ such
  that 
\[
||\breve{\mathbf{u}}{}^{(n)}_\star||_{\bar{\mathcal{S}}_\star,m} < \varepsilon,
\qquad \mbox{and} \qquad
||\breve{\mathbf{u}}{}^{(n)}_\star||_{\bar{\mathcal{S}}_\star,m}
\xrightarrow{n\rightarrow\infty} 0, 
\]
then for the corresponding solutions $\breve{\mathbf{u}}{}^{(n)} \in  C^{m-2}(
[0,2]\times \bar{\mathcal{S}}_\star, \mathbb{R}^N)$, one has $|| \breve{\mathbf{u}}{}^{(n)}||_{\bar{\mathcal{S}}_\star,m} \rightarrow 0$ uniformly in $\tau \in \big{[}  \tau_\star, \frac{5}{2} \big{)}$ as $n \rightarrow \infty$.
\end{itemize}

\end{theorem}

\begin{proof}
The proof is a direct application of Kato's existence, uniqueness and stability
theory for symmetric hyperbolic systems \cite{Kat75b} to developments
with compact spatial sections ---see Theorem 12.4 in
\cite{CFEBook}; see also \cite{MinVal21}.
\end{proof}

\begin{remark}
{\em In view of the localisation properties of hyperbolic equations
  the matching of the perturbation data on $\mathcal{R}_\bullet$ does
  not influence the solution $\mathbf{u}$ on
  $D^+(\mathcal{R}_\bullet)$. Accordingly, in the subsequent
  discussion we discard the solution $\mathbf{u}$ on the region
  $\bar{\mathcal{M}}_\bullet\setminus D^+(\mathcal{R}_\bullet)$ as
  this has no physical relevance. }
  \end{remark}

Moreover, given the \emph{propagation of the constraints}, Lemma
\ref{Lemma:PropagationConstraints}, and the relation between the
extended conformal Einstein field equations and the vacuum Einstein
field equations, Lemma \ref{Lemma:XCFEtoEFE}, one has the following:

\begin{corollary}
\label{Corollary:SolutionsEFE}
The metric 
\[
\bmg= \Theta^{2} \tilde{\bmg}
\]
obtained from the solution to the conformal evolution equations given
in Theorem \ref{Theorem:ExistenceConformalEvolution} implies a
solution to the vacuum Einstein field equations with positive
Cosmological constant on $\tilde{\mathcal{M}}\equiv D^+(\mathcal{R}_\bullet)$. This solution admits a smooth conformal extension
with a spacelike conformal boundary. In particular, the timelike
geodesics fully contained in $\tilde{\mathcal{M}}$ are complete.
\end{corollary}

\begin{remark}
{\em The resulting spacetime $(\tilde{\mathcal{M}},\tilde{\bmg})$
  is a non-linear perturbation of the sub-extremal Schwarzschild-de
  Sitter spacetime on a portion of the Cosmological region of the
  background solution which contains a portion of the asymptotic region.}
\end{remark}

\begin{remark}
{As $\mathcal{R}_\bullet$ is not compact, its development has a Cauchy
  horizon $H^+(\mathcal{R}_\bullet)$. }
\end{remark}

\section{Conclusions}
\label{Section:Conclusions}
This article is a first step in a programme to study the non-linear
stability of the Cosmological region of the Schwarzschild-de Sitter
spacetime. Here we show that it is possible to construct solutions to the
vacuum Einstein field equations in this region containing a portion of
the asymptotic region and which are, in a precise sense, non-linear
perturbations of the exact Schwarzschild-de Sitter
spacetime. Crucially, although the spacetimes constructed have an
infinite extent to the future, they exclude the regions of the
spacetime where the Cosmological horizon and the conformal boundary
\emph{meet}. From the analysis of the asymptotic initial value problem
in \cite{GasVal17a} it is know that the \emph{asymptotic points} in the
conformal boundary from which the horizons emanate contain
singularities of the conformal structure. Thus, they cannot be dealt
by the approach used in the present work which relies on the Cauchy
stability of the initial value problem for symmetric hyperbolic
systems. It is conjectured that the singular behaviour at the
asymptotic points can be studied by methods similar to those used in
the analysis of spatial infinity ---see \cite{Fri98a}. These ideas
will be developed elsewhere.

The next step in our programme is to reformulate the existence and
stability results in this article in terms of a characteristic initial
value problem with data prescribed on the Cosmological horizon. Again,
to avoid the singularities of the conformal structure, the
characteristic data has to
be prescribed away from the asymptotic points. Alternatively, one
could consider data sets which become exactly Schwarzschild-de Sitter
near the asymptotic points. Given the comparative simplicity of the
characteristic constraint equations, proving the existence of such
data sets is not as challenging as in the case of the standard
(i.e. spacelike) constraints. In what respects the evolution problem
it is expected that a generalisation of the methods used in
\cite{HilValZha20b} should allow to evolve characteristics to reach a
suitable hypersurface of constant coordinate $r$. The details of this
construction will be given in a subsequent article.  

\section*{Acknowledgements}
JAVK thanks Volker Schlue for a stimulating conversation on the topic of this
article.

\end{document}